\documentclass[10pt,twocolumn,twoside]{IEEEtran}
\IEEEoverridecommandlockouts
\usepackage{graphicx}
\usepackage{hyperref}
\usepackage{url}
\usepackage{array}
\usepackage{amsmath}
\usepackage{float}
\usepackage{multirow}
\usepackage{array}
\usepackage{cancel}
\usepackage{tabu}
\usepackage{amsthm}
\usepackage{epsfig}
\usepackage{epstopdf}
\usepackage{epsf}
\usepackage{color}
\usepackage[english]{babel}

\theoremstyle{definition}
\newtheorem{theorem}{Theorem}

\newtheorem{remark}{Remark}
\usepackage{amssymb}
\usepackage{cite}
\usepackage{amsmath,amssymb,amsfonts}
\usepackage{graphicx}
\usepackage{textcomp}
\usepackage{xcolor}
\usepackage{algorithm}
\usepackage[noend]{algpseudocode}
\usepackage[T1]{fontenc}
\usepackage{graphicx,lipsum,afterpage,subcaption}
\usepackage{enumitem}
\usepackage{lipsum}
\usepackage{mathtools}
\usepackage{cuted}
\usepackage[thinc]{esdiff}
\DeclareMathOperator*{\argmax}{arg\,max}

\def\BibTeX{{\rm B\kern-.05em{\sc i\kern-.025em b}\kern-.08em
    T\kern-.1em\lower.7ex\hbox{E}\kern-.125emX}}
\newif\ifcomments
\commentstrue
\commentsfalse
\ifcomments
    \usepackage{ulem}
    \def\duccomment#1{{$\!$\color{red} [Thuan: #1]}}
    \def\ducedit#1{{$\!$\color{red} [Thuan: #1]}}

\else 
    \def\duccomment#1{}    
    \def\ducedit#1{}
    
\fi

\begin{document}

	\title
	{Bounded Guaranteed Algorithms for Concave Impurity Minimization Via Maximum Likelihood}

	\author{Thuan Nguyen, Hoang Le and Thinh Nguyen, \textit{Senior Member, IEEE}
\thanks{Thuan Nguyen  is with the Department of Computer Science, Tufts University, Medford, MA, 02155 USA, email: thuan.nguyen@tufts.edu;  Hoang Le and Thinh Nguyen are with the School of  Electrical Engineering and Computer Science, Oregon State University, Oregon, OR,
97331 USA, e-mail: \{lehoang,thinhq\}@oregonstate.edu.
}
\thanks{A part of this work was presented at IEEE International Conference on Acoustics, Speech and Signal Processing (ICASSP) 2021, Toronto, Ontario, Canada  \cite{nguyen2021constant}. In this version, we provided the full proofs for all theorems, introduced two additional algorithms based on the main algorithm in \cite{nguyen2021constant}, connected our results to other well-established results in signal processing and information theory, as well as added more numerical results. }}

\maketitle
\vspace{-0.2 in}
\begin{abstract}
Partitioning algorithms play a key role in many scientific and engineering disciplines.  A partitioning algorithm divides a set into a number of disjoint subsets or partitions. Often, the quality of the resulted partitions is measured by the amount of impurity in each partition, the smaller impurity the higher quality of the partitions.  In general, for a given impurity measure specified by a function of the partitions, finding the minimum impurity partitions is an NP-hard problem.  Let $M$ be the number of $N$-dimensional elements in a set and $K$ be the number of desired partitions, then an exhaustive search over all the possible partitions to find a minimum partition has the complexity of $O(K^M)$ which quickly becomes impractical for many applications with modest values of $K$ and $M$. Thus, many approximate algorithms with polynomial time complexity have been proposed, but few provide bounded guarantee. In this paper, an upper bound and a lower bound for a class of impurity functions are constructed. Based on these bounds, we propose a low-complexity partitioning algorithm with bounded guarantee based on the maximum likelihood principle. The theoretical analyses on the bounded guarantee of the algorithms are given for two well-known impurity functions Gini index and entropy.  When $K \geq N$, the proposed algorithm achieves state-of-the-art results in terms of lowest approximations and polynomial time complexity $O(NM)$. In addition, a heuristic greedy-merge algorithm having the time complexity of $O((N-K)N^2+NM)$ is proposed for $K<N$. Although the greedy-merge algorithm does not provide a bounded guarantee, its performance is comparable to that of the state-of-the-art methods. Our results also generalize some well-known information-theoretic bounds such as Fano's inequality and Boyd-Chiang's bound.
\end{abstract}
\begin{IEEEkeywords}
Partition, quantization, approximation, impurity, entropy, mutual information. 
\end{IEEEkeywords}

\section{Introduction}

Partitioning plays a key role in many scientific and engineering disciplines. It is a key building block in many popular algorithms such as clustering and classification in machine learning.  In signal processing and communications, partitioning algorithms, which are usually called  quantization, aim to minimize the distortion or maximize the mutual information between the original signal and the quantized signals.  A partitioning algorithm divides a set of $M$ $N$-dimensional elements into $K$ disjoint subsets or partitions. Often, the quality of the resulted partitions is measured by the amount of impurity in each partition, the smaller impurity the higher quality of the partitions.  Typically, the amount of impurity is measured  by a real-valued function over the resulted partitions.  

When the elements can be modeled as the outcomes of an underlying probabilistic model, it makes sense to consider some statistical measures such as the average or the variance of the impurity.  Naturally, a partitioning algorithm in this scenario might classify the elements into different partitions using probability distributions, rather than the values of the elements.  For example, let us consider a popular impurity function using the Shannon entropy \cite{quinlan2014c4}, \cite{coppersmith1999partitioning}, \cite{chou1991optimal}. Consider a set whose elements are outcomes of a random variable $W$.  A large entropy of a random variable implies that the elements are likely to be different, i.e., the set has a high level of non-homogeneity or "impurity".  A $K$-optimal partition algorithm divides the original set into $K$ subsets such that the weighted sum of entropies in each subset is minimum.  Since entropy is a concave function of the probability mass function andnot the values of a random variable, the partition algorithms, in this case, work directly with the $N$-dimensional probability mass vector.  In contrast, the popular $k$-means algorithms do not assume an underlying probabilistic model for how the elements come about.  Thus, the elements are clustered using a distance measure (typically Euclidean) which is a function of the actual values of the elements.  

In general, for a given impurity measure specified by a function over the partitions, finding the minimum impurity partitions is an NP-hard problem.  Since the number of possible partitions is $K^M$, an exhaustive search over all the possible partitions to find a minimum partition has the complexity of $O(K^M)$ which quickly becomes impractical for many applications with modest values of $K$ and $M$.   
To that end, many approximate algorithms with polynomial time complexity have been proposed, but few provide bounded guarantee  \cite{quinlan2014c4},  \cite{coppersmith1999partitioning}, \cite{burshtein1992minimum}, \cite{breiman2017classification}.  Many of these algorithms exploit the concavity of the impurity function to reduce the time complexity.  For example, in  \cite{breiman2017classification}, an algorithm is proposed to find the optimal partition using a concave impurity function with the computational complexity of $O(M \log M)$ for binary classification tasks ($K=2$). Burshtein \textit{et al.} \cite{burshtein1992minimum} and  Coppersmith \textit{et al.} \cite{coppersmith1999partitioning} provided algorithms and theoretical analysis for the partitioning problem for a general concave impurity function called "frequency-weighted impurity". These "frequency-weighted impurity" are concave functions over its second argument.  Two popular impurity functions the Gini index \cite{breiman2017classification} and Shannon entropy \cite{quinlan2014c4} belong to this class of frequency-weighted impurity.  Burshtein \textit{et al.} and Coppersmith \textit{et al.} showed that an optimal frequency-weighted impurity partition is separated by hyperplane cuts in the space of probability distributions.  Based on this insight, they also proposed polynomial time algorithms to determine the optimal partitions \cite{coppersmith1999partitioning}, \cite{burshtein1992minimum}, \cite{nguyen2020linear}. Based on the work of Burshtein \textit{et al.}, Kurkoski and Yagi proposed an algorithm to find the globally optimal partition that minimizes entropy impurity in $O(M^3)$ when $N=2$ \cite{kurkoski2014quantization}. 

Although many heuristic algorithms have been proposed, there are few results that provide  bounded performance guarantees.  Murtinho and Laber \cite{laber2019minimization} presented a (1 + $\epsilon$) approximation algorithm for Gini index in O($2^{K \log N}$) time complexity.  Laber \textit{et al.} \cite{laber2018binary}  constructed  a 2-approximation algorithm with computational complexity of $O(2^N  M \log M)$ for binary partition ($K=2$).  In other words, the impurity resulted from the proposed algorithm is at most a factor of 2 away from the smallest impurity.  The complexity can be  further reduced to $O(MN+ M \log M)$ at the expense of increasing the approximation factor from 2 to $3+\sqrt{3}$.
We also note that the algorithm in \cite{laber2018binary} is closely related to the well-established
Twoing method in \cite{breiman2017classification}.  Moreover, the application of the algorithm in \cite{laber2018binary} is limited to binary partitions ($K=2$).  Cicalese \textit{et al.} \cite{cicalese2019new} extended the work in \cite{laber2018binary} with a heuristic algorithm for $K > 2$. The algorithm can achieve $\log^2 (\min\{N,K\})$-approximation for the entropy impurity,  3-approximation for 
the Gini index impurity if $K<N$ and 2-approximation for 
the Gini index impurity if $K \geq N$.  It is the first constant factor algorithm for clustering based on minimizing entropy impurity that does not rely on assumptions about the input data.  Our analysis in Appendix \ref{apd: time complexity analysis} shows that the complexity of the algorithm in \cite{cicalese2019new} is reduced to polynomial time because the most time consuming step of the algorithm is based on the dynamic programming technique in \cite{kurkoski2014quantization} which reduces the time complexity to $O(M^3)$. Using the SMAWK algorithm \cite{aggarwal1987geometric}, we showed that the dynamic programming step of the algorithm in \cite{cicalese2019new} can be further reduced to $O(M \log M )$. The analysis of the algorithm in \cite{cicalese2019new} together with a suggested method for reducing the computational complexity is shown in Appendix \ref{apd: time complexity analysis}. 

The partitioning algorithm also tightly relates to clustering algorithms which group $P$
probability distributions into $K$ clusters in such a way to minimize a certain distance. For example, minimizing entropy impurity partition is equivalent to finding the optimal clusters that minimizes the Kullback-Leibler (KL) distance \cite{zhang2016low}.  Generally,  the local optimal solution minimizing impurity partition can be found based on the famous $k$-means algorithm with a suitable distance \cite{coppersmith1999partitioning}, \cite{nguyen2020linear}, \cite{zhang2016low}. 
Thus, the results about  approximation for clustering with KL-divergence in \cite{sra2008approximation}, \cite{chaudhuri2008finding} can be applied to find a good partition that minimizes 
the entropy impurity. For example, in \cite{sra2008approximation}, Sra et. al. showed that a $k$-means algorithm using the KL-divergence distance metric with an exponential time worst-case complexity (see \cite{vattani2011k}) can obtain $O(\log K)$-approximation of the optimal clustering.  The algorithm of Chaudhuri and McGregor \cite{chaudhuri2008finding} can provide $\log(M)$-approximation for finding a good clustering in polynomial time complexity. On the other hand, the quality of the approximation algorithm  in \cite{chaudhuri2008finding} depends on the size of the dataset.  In many settings where $M$ (number of data points) tends to be large while $N$ (data dimension) and $K$ (number of clusters) tend to be much smaller, thus the algorithm proposed by Cicalese \textit{et al.} \cite{cicalese2019new} is useful due to a smaller constant factor approximation of $\log^2 (\min\{N,K\})$ as compared to $\log(M)$ in \cite{chaudhuri2008finding}.

The contributions of this paper are fivefold: 
\begin{itemize}

\item   We construct an upper bound and a lower bound for a class of impurity functions including both the Gini index and entropy (see Theorem 1 and Theorem 4). To minimize the gap between these bounds, an approximation algorithm based on the maximum likelihood principle is introduced. The proposed algorithm (Algorithm \ref{alg: finding emax}) can provide comparable theoretical performance as well as numerical performances to the state-of-art methods in \cite{cicalese2019new}. Particularly, our theoretical bound is $2$-approximation\footnote{An $\alpha$-approximation algorithm is an algorithm that provides a solution no larger than $\alpha$ times the global optimal solution. For example, a $2$-approximation algorithm will output a partition that has the impurity at most two times the impurity introduced by the global optimal partition.} for the Gini index (see Theorem 7) and $\log^2 N$-approximation for entropy impurity under a certain condition (see Theorem 8). The time complexities of the proposed Algorithm \ref{alg: finding emax} is $O(NM)$ if $K \geq N$ and is $O(2^{N/2} NM)$ if $K < N$, respectively. 

\item  Based on the Algorithm \ref{alg: finding emax}, we propose a so-called greedy-splitting algorithm (Algorithm \ref{alg: greedy splitting}) to achieve a better splitting quality when $K>N$. Greedy-splitting algorithm still runs in $O(KNM)$ and achieves the bound at least equal or better compared to the bound of the original maximum likelihood algorithm e.g., $2$-approximation for Gini index and $\log^2 N$-approximation for entropy impurity. When $K<N$, the proposed Algorithm \ref{alg: finding emax} runs in $O(2^{N/2} NM)$ which is exponential in term of $N$. To reduce the time complexity, we proposed a so-called greedy-merge algorithm (Algorithm \ref{alg: greedy merge}) having the time complexity of $O((N-K)N^2 +NM)$ which is linear in the size of the dataset $M$. Although the greedy-merge algorithm does not provide a guarantee on splitting quality, it shows a comparable performance to the results provided by the proposed algorithm in \cite{cicalese2019new}. 

\item  In contrast to the state-of-the-art methods that can only handle a specific type of impurity function such as the Gini index or entropy \cite{laber2019minimization} \cite{cicalese2019new}, the proposed lower bound and upper bound are constructed using elementary techniques, and therefore, can be adapted to a wide class of concave impurity functions. 

\item We suggest a method that can improve the complexity of the algorithm in \cite{cicalese2019new} from $O(NM+M^3)$ down to $O(NM+M \log M )$ based on the matrix searching SMAWK algorithm \cite{aggarwal1987geometric}.

\item Our technique on bounding the solution confirms and generalizes well-established results in signal processing and information theory.  In particular, the Fano's inequality \cite{cover2012elements}  and Boyd-Chiang's upper bound \cite{chiang2004geometric} for channel capacity are consequences of our results.
\end{itemize}

Even though our algorithms and approximation factors share some similarities with the results in \cite{cicalese2019new} (see Remark 3, 5, 6, and 7), the main difference comes from our elementary proofs that are mainly constructed based on the lower bound and the upper bound for a wide class of impurity functions. These elementary techniques enable our approach to be easily adapt to any concave impurity functions.

From signal processing, communication, and information theory's perspectives, our work is related to optimal quantization design for constructing polar code \cite{tal2013construct} and low density parity code (LDPC) decoder \cite{romero2015decoding}.  These optimal quantizers aim to maximize the mutual information between input and output \cite{kurkoski2014quantization}, \cite{zhang2016low}, \cite{nguyen2021minimizing, nguyen2020communication, mathar2013threshold, nguyen2018capacities, kurkoski2017single, vangala2015quantization}. Kurkoski \textit{et al.} \cite{kurkoski2014quantization} showed that for a given input distribution, finding an optimal quantizer that maximizes the mutual information is equivalent to finding an optimal partition that minimizes the entropy impurity.  Thus, our algorithm can be applied to find a good quantizer that maximizes mutual information. 
In addition, the problem of minimizing impurity partitions also relates to the well-known Information Bottleneck Method (IBM) \cite{tishby2000information} and Deterministic Information Bottleneck (DIB) \cite{strouse2017deterministic}. Particularly, both IBM and DIB can be viewed as the problems of minimizing entropy impurity partition under constraints for a given input distribution.  Thus, the results in this paper can be used to design good approximation algorithms for IBM and DIB. 

The outline of our paper is as follows.  In Section \ref{sec: problem} describes the problem formulation.  In Section \ref{sec: solution}, an upper bound of impurity partition is constructed together with an algorithm that provably achieves near-optimal impurity. The proof of near-optimal partition together with a lower bound of impurity function is characterized in Section \ref{sec: proof of near optimal}. To reduce the time complexity of the proposed algorithm in Section \ref{sec: solution}, we propose two greedy algorithms having linear time complexity in Section \ref{sec: greedy algs}. The numerical results are provided in Section \ref{sec: simulation}.  Finally, we provide a few concluding remarks in Section \ref{sec: conclusion}.

\section{Problem Formulation}
\label{sec: problem}
\begin{figure}
  \centering
  \includegraphics[width=1.5 in, height = 1.2 in]{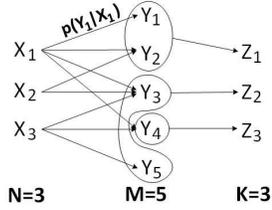}\\
  \caption{Finding an optimal quantizer $Q^*(Y) \rightarrow Z$ such that $I_{Q^*}$  is minimized.}\label{fig: 1}
 \end{figure} 
\subsection{Problem formulation}
We assume that the set $\mathbb{Y}$ to be partitioned consists of $M$ discrete data points generated from an underlying probabilistic model.  Specifically, let $X$ be a discrete random variable taking on the values $x_1, x_2, \dots, x_N$ with a given probability mass vector $\mathbf{p}_{\mathbf{x}} = (p(x_1), p(x_2), \dots, p(x_N))$.  Let $Y$ be another discrete random variable taking on values $y_1, y_2, \dots, y_M$ which follows a given conditional probability $p(y_j|x_i)$.   The goal is to partition $\mathbb{Y}$ into $K$ partitions to minimize a given impurity function over the resulted partitions.  
Fig. \ref{fig: 1} shows a generative model for $Y$. $Y$ is then partitioned/quantized to $Z$ using a partition scheme/quantizer $Q$.

$$Q(Y) \rightarrow Z.$$

$Z$ is modeled as a discrete random variable $Z$ taking on values $z_1, z_2, \dots, z_K$.  In this setting, for given $p(x_i)$ and $p(y_j|x_i)$,  $p(x_i, y_j)$ are assumed to be given $\forall i,j$. Thus, each data point $y_j$ is represented by a joint distribution vector $\mathbf{p}_{\mathbf{x},y_j}=(p(x_1,y_j), p(x_2,y_j), \dots, p(x_N,y_j))$.
Each quantizer $Q$ induces  a joint distribution vector  $\mathbf{p}_{\mathbf{x},z_k}=(p(x_1, z_k),p(x_2, z_k),\dots, p(x_N, z_k))$ between $X$ and $Z=z_k$. The conditional distribution $p(x_i|z_k)$ of $X$ given $Z$ and  the marginal probability mass function $p(z_k)$ of $Z$ can be determined from $p(x_i, z_k)$.  We want to find an optimal quantizer $Q^*$ that minimizes the impurity function $I_Q$ that satisfies two following conditions:

\begin{itemize}
\item (\textbf{Required}) $I_Q$ has the following form:
\begin{equation}
\label{eq:impurity}
I_Q=\sum_{k=1}^{K}\sum_{i=1}^{N}{p(z_k)f(p(x_i|z_k))},
\end{equation}		
where $f(.): \mathbb{R} \rightarrow \mathbb{R}^+$ is a non-negative concave function.  $f(x)$ is concave over a continuous interval $\mathbb{S}$ if for any $a, b \in \mathbb{S}$,    \begin{equation}
\label{eq: concave function}
f(\lambda a + (1-\lambda)b) \geq \lambda f(a) + (1-\lambda)f(b), \forall \lambda \in (0,1).
\end{equation}

We note that $\sum_{i=1}^{N}f(p(x_i|z_k))$ in (\ref{eq:impurity}) is the impurity contribution from partition $z_k$.  Therefore,  $I_Q$ is viewed as the weighted average impurity over all the partitions.  Many well-known impurity functions such as Gini index \cite{quinlan2014c4} and entropy \cite{breiman2017classification} have concave $f(.)$.

\item (\textbf{Optional}) $f(x) = xl(x)$ where $l(x):\mathbb{R} \rightarrow \mathbb{R}$ is a convex function.  $l(x)$ is convex over a continuous interval $\mathbb{S}$ if for any $a, b \in \mathbb{S}$,    \begin{equation}
\label{eq: convex function of l(x)}
l(\lambda a + (1-\lambda)b) \leq \lambda l(a) + (1-\lambda)l(b), \forall \lambda \in (0,1).
\end{equation} 

This second condition is optional in the sense that we only use it in the analysis of the approximation of the proposed algorithm.  The algorithm itself does not make use of this condition.  Furthermore, many popular impurity functions indeed satisfy this second condition.
\end{itemize}

{\bf Examples of popular impurity functions:}

\begin{itemize}
\item{} Entropy function:
Let $f(x) = -x \log{x}$ which can be shown to be a concave function.  Replacing $f(x) = -x \log{x}$ with $x = p(x_i|z_k)$ into (\ref{eq:impurity}), we have:
\begin{equation}
I_Q = \sum_{k=1}^{K} \sum_{i=1}^{N}p(z_k)[-p(x_i|z_k)\log{(p(x_i|z_k))}],
\end{equation}
which is the weighted conditional entropy of $X$ given $Z$.
Also let $l(x) = -\log{x}$,  $l(x)$ is a convex function. Thus $f(x) = -x\log{x}$ satisfies  the second optional condition.

\item{} Gini index function: Given a set $A$ of elements with random $N$ labels according to the distribution of the labels $\mathbf{p}_\mathbf{A} = (p(A_1), p(A_2), \dots, p(A_N))$.  The Gini impurity is a measure of how often a randomly chosen element based on the label distribution would be mislabeled.  Specifically,  since the probability of picking an element with the label $i$ is $p(A_i)$ then the probability of mislabeling that element is $\sum_{l \neq i}{p(A_i)} = 1 - p(A_i)$.  Summing all $i$, the probability of mislabeling an element is: $$\sum_{i=1}^{N}p(A_i)(1-p(A_i)).$$ 
Let $f(x)=x(1-x)$ which can be shown to be a concave function. Replacing $f(x)=x(1-x)$ using $x=p(x_i|z_k)$ into (\ref{eq:impurity}), 
the Gini index impurity \cite{coppersmith1999partitioning} has the following form:
\begin{eqnarray}
I_Q =\sum_{k=1}^{K} \sum_{i=1}^{N} p(z_k) [p(x_i|z_k) (1-p(x_i|z_k))].
\end{eqnarray}
Additionally, let $l(x)=1-x$, $l(x)$ is a linear function, therefore, $l(x)$ is a convex function. Thus the Gini index impurity satisfies  the second optional condition.
\end{itemize}

We note that in \cite{cicalese2019new} and \cite{laber2018binary}, to guarantee their  approximations, the authors considered a class of impurity concave functions $f(.)$ with an additional condition  on $xf''(x)$ being a non-increasing function.

\section{Impurity Minimization Algorithm}
\label{sec: solution}

In this section, we first construct both upper and lower bounds for impurity functions of the form in (\ref{eq:impurity}).  Using these bounds, we show that the proposed maximum likelihood algorithm can achieve a good approximation. In other words, the resulted solution is guaranteed to be away from the true solution by at most a factor that does not depend on the number of data points $M$.  

We define three important quantities below:
\begin{equation}
\label{eq:i^*}
k^* =\argmax_{ 1 \leq i \leq N}{ p(x_i|z_k)},
\end{equation}
\begin{equation}
e_Q=\sum_{k=1}^{K}{p(z_k)p(x_{k^*}|z_k)},
\end{equation}
and 
\begin{equation}
e^{\max}=\max_Q e_Q.
\end{equation}
For a given $k$, $z_k$ is most likely be produced by $x_{k^*}$. 
Therefore, $e_Q$ is the weighted sum of the maximum likelihood of each $x_{k^*}$ for each $z_k$.  
We note that each partition scheme/quantizer $Q$ induces a $p(x_i, z_k)$ and thus $p(x_i|z_k)$.  So $k^*$  and $e_Q$ are different for different $Q$. Our approach to find the minimum impurity is to find two functions: $u(e_Q)$ and $l(e_Q)$ such that  $l(e_Q) \leq I_Q \leq u(e_Q)$.  Furthermore, we show that $u(e_Q)$ and $l(e_Q)$ are decreasing functions for many impurities.  Therefore, by minimizing $u(e_Q)$, i.e., maximizing $e_Q$, we can bound the minimum value of $I_Q$ between $u(e_Q)$ and $l(e_Q)$ for some $e_Q$.   

\subsection{Upper Bound  of The Impurity Function}
We have the following theorem for the upper bound of an impurity function $I_Q$. 
\begin{theorem} ({\bf Upper bound})
\label{theorem: 1}
For any given quantizer $Q$ that induces the corresponding $p(x_i|z_k)$ and: 
\begin{equation}
e_Q=\sum_{k=1}^{K}p(z_k)p(x_{k^*}|z_k),
\end{equation}
let:
\begin{equation}
    u(e_Q)=f(e_Q)+ (N-1)f(\dfrac{1-e_Q}{N-1}),
\end{equation}
then $\forall e_Q$, we have:
\begin{equation}
u(e_Q) \ge I_Q.
\end{equation}
\end{theorem}
\begin{proof}
From the definition of the impurity function, we have:
\begin{small}
\begin{eqnarray}
I_Q \! &\!=\!& \!\sum_{k=1}^{K}\sum_{i=1}^{N} p(z_k)f(p(x_i|z_k)) \nonumber \\
&\!=\!& \! \sum_{k=1}^{K}p(z_k)f(p(x_{k^*}|z_k)) + \sum_{k=1}^{K}\sum_{i\neq k^*, i=1}^{N} p(z_k)f(p(x_i|z_k)) \nonumber \\
&\!\leq\!& \! f(\sum_{k=1}^{K} \! p(z_k)p(x_{k^*}|z_k))\!+\! \sum_{k=1}^{K} \! \sum_{i\neq k^*, i\!=\!1}^{N} \! p(\!z_k\!) \! f(p(x_i|z_k)) \label{eq: 14}\\
&\!\leq\!& f(\sum_{k=1}^{K} p(z_k)p(x_{k^*}|z_k)) \nonumber\\
&& \!+\! \sum_{k=1}^{K}p(z_k) [(N\!-\!1) f(\dfrac{\sum_{i\!=\!1,i \! \neq \! k^*}p(x_i|z_k)}{N-1})]\label{eq: 15}\\
&=&f(e_Q) + (N-1)\sum_{k=1}^{K}p(z_k)  f(\dfrac{1-p(x_{k^*}|z_k)}{N-1}) \label{eq: 16}\\
&\!\leq\!& \! f(e_Q)+(N-1) f(\dfrac{\sum_{k=1}^{K}p(z_k)(1-p(x_{k^*}|z_k))}{N-1}) \label{eq: 17}\\
&\!=\!& f(e_Q)+ (N-1) f(\dfrac{1-e_Q}{N-1}), \label{eq: 18}
\end{eqnarray}
\end{small}where (\ref{eq: 14}) is due to concavity of $f(.)$ and $\sum_{k=1}^{K}p(z_k)=1$, (\ref{eq: 15}) is due to Jensen inequality for concave function (please see the Appendix \ref{apd: jensen inequality}), (\ref{eq: 16}) is due to the definition of $e_Q$ and $\sum_{i=1,i \neq k^*}^{N}p(x_i|z_k)+ p(x_{k^*}|z_k)=1$, (\ref{eq: 17}) is due to concavity of $f(.)$ and $\sum_{k=1}^{K}p(z_k)=1$, (\ref{eq: 18}) is due to $\sum_{k=1}^{K}p(z_k)=1$. 
\end{proof}
\begin{remark}({\bf Fano's inequality.})
There is an interesting connection between $u(e_Q)$ and the well-known Fano's inequality from the information theory.
Specifically, if $f(x)$ is the entropy function, then the upper bound in Theorem \ref{theorem: 1} is identical to the  Fano's inequality. Please see the details of the derivations in the Appendix \ref{apd: fano}.
\end{remark}

\begin{remark}({\bf Maximum likelihood decoding.})
Consider a communication setting with $X$ and $Z$ being the two random variables that represent the transmitted symbols and the received symbols, respectively.  The goal for a receiver is to recover $X$ based on $Z$.   A maximum likelihood decoder  maximizes the posterior probability of $X$ given $Z$.  Specifically, if a symbol $z_k$ is received, then the transmitted symbol is $x_{k^*}$ where $k^*=\argmax_{1\leq i \leq N} p(x_i|z_k)$.  Consequently, $e_{Q}= \sum_{k=1}^{K}{p(z_k)p(x_{k^*}|z_k)}$
is the probability of decoding a transmitted symbol correctly using the mapping $Q$,
and $P_e =  1-e_{Q}$
is the probability of decoding a transmitted symbol incorrectly.
\end{remark}

\begin{theorem}
\label{theorem: monotonic decreasing of upper bound}
$u(e_Q)$ is a monotonic decreasing function.
Moreover, $I_Q = u(e_Q)$ when  $e_Q=\dfrac{1}{N}$ or $e_Q=1$.
\end{theorem}
\begin{proof}
Please see the Appendix \ref{apd: proof of theorem monotonic upper bound}. 
\end{proof}
Based on Theorem \ref{theorem: monotonic decreasing of upper bound}, let $e^{\max}$ be the maximum value over all $e_Q$ i.e., $e^{\max} = \max_Q{e_Q}$, then $u(e^{\max})$ has the minimum value. Since $u(e_Q)$ is an upper bound of $I_Q$, $u(e^{\max})$ provides a good upper bound for $I_{Q^*}$.
We now state an important result for a special case where the sample space of $Z$ is identical to that of $X$.  In other words, $K=N$ and $z_k = x_i$ for some $k$ and $i$. 

\begin{theorem}({\bf Structure of the $\mathbf{e^{\max}}$ quantizer})
\label{theorem: p(x,y)}
Let $\mathcal{Z}$ and $\mathcal{X}$ be the sample spaces of $Z$ and $X$, respectively. 
Let $j^* = \argmax_i p(x_i, y_j)$  and define quantizer $Q_{e^{\max}}$  with the following structure:
\begin{equation}
\label{eq: optimal quantizer structure}
Q_{e^{\max}}(y_j) = z_{j^*}.
\end{equation}

(a) If $|\mathcal{Z}| = |\mathcal{X}|$, then $Q_{e^{\max}}$ produces $e^{\max} = \max_Q{e_Q}$.  Conversely, for any $Q$ that produces $e^{\max}$, $Q$ must have the structure of $Q_{e^{\max}}$. 

(b) If $|\mathcal{Z}| > |\mathcal{X}|$, then $Q_{e^{\max}}$ still produces $e^{\max} = \max_Q{e_Q}$. However, it is not necessary that for any $Q$ that produces $e^{\max}$, $Q$ must have the structure of $Q_{e^{\max}}$.
\end{theorem}
\begin{proof}
Please see the Appendix \ref{apd: proof of theorem {theorem: p(x,y)}}.
\end{proof}

We note that $j^*$ takes on values $1, 2, \dots, N$, and $z_{j^*}$'s represent the $K=N$ partitions. In other words, when $K > N$ then existing an optimal quantizer $Q_{e^{\max}}$ that produces exactly $N$-partition rather than $K$-partition. Interestingly, for $K>N$, the mapping using only $N$-partition in Theorem \ref{theorem: p(x,y)}-(b) is still optimal i.e., it produces the partitions achieving $e^{\max}$. However, Theorem \ref{theorem: p(x,y)}-(b) does not guarantee any $Q$ that produces $e^{\max}$ must have the structure of $Q_{e^{\max}}$. Indeed, there might exist other quantizers that achieve $e^{\max}$. On the other hand, Theorem \ref{theorem: p(x,y)}-(a)  states that if $K=N$, then any quantizer producing $e^{\max}$ must have the structure of $Q_{e^{\max}}$ in (\ref{eq: optimal quantizer structure}). This necessary condition helps to find $Q_{e^{\max}}$ when $K<N$ as to be shown later. The detail of proof is in Appendix \ref{apd: proof of theorem {theorem: p(x,y)}}. 

\subsection{Algorithm}
Based on the upper bound in Theorem \ref{theorem: 1}, to minimize the impurity function, one wants to minimize the impurity's upper bound $u(e_Q)$. Based on Theorem \ref{theorem: monotonic decreasing of upper bound}, to minimize $u(e_Q)$, one wants to maximize $e_Q$.  To maximize $e_Q$, we propose the algorithm below which utilizes the result of Theorem \ref{theorem: p(x,y)}.

Let $\mathcal{V}_K$ be the set of binary $N$-dimensional vectors $\mathbf{v}$'s, each contains exactly $K$ entries 1 and  $N-K$ entries 0. Thus, the size of $\mathcal{V}_K$ is ${N \choose K}$.  For each $\textbf{v} = (v_1, v_2, \dots, v_N)$, define the $N$-dimensional vector:
\begin{align*}
    \mathbf{p}'_{\mathbf{x},y_j} &=  \big[ v_1p(x_1,y_j),v_2 p(x_2,y_j), \dots, v_Np(x_N,y_j) \big] \\
    &=\big[p'(x_1,y_j), p'(x_2,y_j), \dots, p'(x_N,y_j) \big]
\end{align*} 
then $\mathbf{p}'_{\mathbf{x},y_j}$ has exactly $K$ non-zero entries. Next, we consider the following  possible cases. 
\begin{itemize}

\item \textbf{$K = N$}:    When $K=N$, $\mathcal{V}_K=\mathcal{V}_N$ contains exactly one $\mathbf{v}$ which is $\mathbf{v}=(1,1,\dots,1)$.  In this case, $p'(x_i,y_j)$ = $p(x_i,y_j)$. Thus, using Theorem \ref{theorem: p(x,y)}-(a) with $p(x_i,y_j)$ replaced by $p'(x_i,y_j)$ will produce $e^{\max}$.  

\item \textbf{$K < N$}: When $K < N$, there are ${N \choose K}$ quantizers $Q$ that partition $K$-dimension vectors $\mathbf{p}'_{\mathbf{x},y_j}$ to $K$ partitions. Moreover, from the necessary condition in Theorem \ref{theorem: p(x,y)}-(a), at least one of quantizer in this ${N \choose K}$ quantizers must achieve $e^{\max}$.

\item \textbf{$K>N$}: From  Theorem \ref{theorem: p(x,y)}-(b),  the partition which achieves $e^{\max}$ is exactly the same with the partition when $K=N$. In other words, the partition can be achieved using the maximum likelihood principle using  $\mathbf{v}=(1,1,\dots,1)$, and the optimal partitions which produces $e^{\max}$ has $N$ nonempty partitions together with $K-N$ empty partitions.
\end{itemize}

Based on three possible cases above, the maximum  likelihood algorithm (Algorithm \ref{alg: finding emax}) follows. The detail of the proof is shown in Appendix \ref{apd: proof of theorem {theorem: p(x,y)}}. 

\begin{footnotesize}
\begin{algorithm}
\caption{ Finding $e^{\max}$ Algorithm.}
\label{alg: finding emax}
\begin{algorithmic}[1]
\State{\textbf{Input}: Dataset $Y=\{y_1,\dots,y_M\}$ and  $p(x_i,y_j)$, $K$ and $N$}.

\State{\textbf{Output}: Partition $Z=\{z_1,z_2,\dots,z_K\}$}.  

\State{\textbf{If $K<N$: $\mathcal{V} = \mathcal{V}_K$}}

\State{\textbf{If $K \geq N$: $\mathcal{V} = \mathcal{V}_N$}}

\State{\textbf{\hspace{0.2in} For $\textbf{v} \in \mathcal{V}$}}

\State{\hspace{0.4in} \textbf{For $1 \leq j \leq M$, $1 \leq i \leq N$}}

\State{\hspace{0.6 in} \textbf{Step 1}: Projection.}
\begin{equation}
 p'(x_i,y_j)   = v_ip(x_i,y_j).
\end{equation}

\State{\hspace{0.6 in} \textbf{Step 2}: Finding the maximum likelihood.}
\begin{equation}
\label{eq: maximum likelihood-2}
j^*=\argmax_{1\leq i \leq N} \{ p'(x_i,y_j) \}.
\end{equation}
\State{\hspace{0.6 in} \textbf{Step 3}: Partition assignment.}
\begin{equation}
\label{eq: partition assignment-2}
Q(y_j) \rightarrow z_{j^*}.
\end{equation}

\State{\hspace{0.4 in}\textbf{End For }}

\State{\hspace{0.4 in}\textbf{Computing $e_Q$:} Using the resulted partitions to compute $e_Q$.}

\State{\textbf{ \hspace{0.2 in} End For }}

\State{\textbf{Return:} Returning the partition that produces $e^{\max}=\max_Q e_Q$.}
\end{algorithmic}
\end{algorithm}
\end{footnotesize}

\textbf{Time complexity of Algorithm \ref{alg: finding emax}:} To find the partition that generates $e^{\max}$, we need to search over all the possible mappings  $\textbf{v} \in \mathcal{V}_K$. For each $\mathbf{v}$,  Algorithm \ref{alg: finding emax} has complexity of $O(NM)$.  Since there are ${N \choose K}$ possible $\mathbf{v}$ if $K<N$, Algorithm \ref{alg: finding emax} has the complexity of
$O({N\choose K} NM)$. In the worst case when $K = N/2$, we have ${N \choose N/2} = 2^{N/2}$ and the complexity of Algorithm \ref{alg: finding emax} is $O(2^{N/2} NM)$. However, if $K \geq N$,  there is only one mapping $\mathbf{v}$ and the time complexity of algorithm is truly in linear of $O(NM)$.

\begin{remark}
The idea of Algorithm \ref{alg: finding emax} is mainly based on the maximum likelihood principle and also shares some similarities with the Dominance algorithm in \cite{cicalese2019new}. Particularly, the proposed Algorithm \ref{alg: finding emax} is similar to the Dominance algorithm when $K \geq N$, however, these two algorithms are different if $K < N$. We refer the readers to Sec. 4 of \cite{cicalese2019new} for more details how Cicalese \textit{et al.} handled their partitions when $K < N$.  
\end{remark}


\section{Approximation Analysis for Entropy and Gini Index}
\label{sec: proof of near optimal}

In this section, we state a few results for establishing the approximation property of Algorithm \ref{alg: finding emax}. First, the following theorem establishes a lower bound for $I_Q$.  

\begin{theorem}({\bf Lower bound})
\label{theorem: 2}
For any given quantizer $Q$ that induces the corresponding $p(x_i|z_k)$ and:
\begin{equation}
e_Q=\sum_{k=1}^{K}p(z_k)p(x_{k^*}|z_k),
\end{equation}
then $\forall$$e_Q$, we have:
\begin{equation}
    I_Q \geq l(e_Q).
\end{equation}
\end{theorem}

\begin{proof}
Using the concavity definition of $f(x)$ in (\ref{eq: concave function}),  and let $t$ and $q$ be the positive scalars such that $0 \leq t \leq q$, we have:

\begin{eqnarray}
   f(t) \geq  (1-\dfrac{t}{q})f(0)+\dfrac{t}{q} f(q)
   =\dfrac{t}{q} f(q).  \label{eq: 21} 
\end{eqnarray}
From the definition of the impurity function, we have:
\begin{small}
\begin{eqnarray}
I_Q&=&\sum_{k=1}^{K}\sum_{i=1}^{N}{p(z_k)f(p(x_i|z_j))} \nonumber\\
&=&\sum_{k=1}^{K}p(z_k) f(p(x_{k^*}|z_k)) \\&+& \sum_{k=1}^{K}p(z_k) \Big(\sum_{i\neq k^*, i=1}^{N}f(p(x_{k^*}|z_k))\Big) \nonumber\\
&\geq&\sum_{k=1}^{K}p(z_k) f(p(x_{k^*}|z_k)) \\&+&\sum_{k=1}^{K}p(z_k) \Big(\sum_{i\neq j^*, i=1}^{N} \dfrac{p(x_i|z_k)}{p(x_{k^*}|z_k)}f(p(x_{k^*}|z_k))\Big) \label{eq: 22}\\
&=&\sum_{k=1}^{K}p(z_k) f(p(x_{k^*}|z_k)) \\&+& \sum_{k=1}^{K}p(z_k) \dfrac{\sum_{i\neq k^*, i=1}^{N} p(x_i|z_k)}{p(x_{k^*}|z_k)}f(p(x_{k^*}|z_k)) \label{eq: 23}\\
&=&\sum_{k=1}^{K}p(z_k) f(p(x_{k^*}|z_k)) \\&+& \sum_{k=1}^{K}p(z_k)\Big(\dfrac{1-p(x_{k^*}|z_k)}{p(x_{k^*}|z_k)}\Big)f(p(x_{k^*}|z_k)) \label{eq: 24}\\
&=&\sum_{k=1}^{K}p(z_k) f(p(x_{k^*}|z_k)) \Big(1+ \dfrac{1-p(x_{k^*}|z_k)}{p(x_{k^*}|z_k)}\Big) \label{eq: 25-a}\\
&=&\sum_{k=1}^{K}p(z_k) f(p(x_{k^*}|z_k))\dfrac{1}{p(x_{k^*}|z_k)} \label{eq: 25}\\
&=&\sum_{k=1}^{K}p(z_k) p(x_{k^*}|z_k) l(p(x_{k^*}|z_k))\dfrac{1}{p(x_{k^*}|z_k)} \label{eq: 26}\\
&=&\sum_{k=1}^{K}p(z_k)  l(p(x_{k^*}|z_k)) \label{eq: 27}\\
&\geq& l(\sum_{k=1}^{K} p(z_k) p(x_{k^*}|z_k)) \label{eq: 28}\\
&=& l(e_Q),
\end{eqnarray}
\end{small}with (\ref{eq: 22}) is due to (\ref{eq: 21}) using $t= p(x_i|z_k)$ and $q=p(x_{k^*}|z_k)$ and noting that $p(x_{k^*}|z_k) \geq p(x_i|z_k)$ $\forall i$, (\ref{eq: 26}) is due to $f(x)=xl(x)$, and (\ref{eq: 28}) due to the Jensen inequality for the convex function $l(x)$. The lower bound is tight i.e., $I_Q=l(e_Q)$ if $e_Q=\dfrac{1}{N}$ or $e_Q=1$.
\end{proof}
\begin{remark}
There is a connection between the lower bound above and the well-known Boy-Chiang upper bound of channel capacity. 
Specifically, for a uniform input distribution, if $f(x)$ is the entropy function, then the lower bound in Theorem \ref{theorem: 2} implies the Boy-Chiang upper bound of channel capacity \cite{chiang2004geometric}. More details  are in the Appendix \ref{apd: boyd-chiang}. \\
\end{remark}

\begin{theorem}({\bf $\mathbf{R(e^{\max})}$-approximation})
\label{theorem: near opt}
Algorithm \ref{alg: finding emax} provides $R(e^{\max})$-approximation for both entropy and Gini index impurities where:
\begin{eqnarray}
\label{eq: remax definition}
R(e^{\max})=\dfrac{u(e^{\max})}{l(e^{\max})}.
\end{eqnarray}
\end{theorem}

\begin{proof}
Let $I_{Q^*}$ be the minimum impurity and $I_{Q_{e^{\max}}}$ be the impurity of the partition produced by running Algorithm \ref{alg: finding emax}.
Now, assume that $Q^*$ produces $e_{Q^*}$. From the definition of $e^{\max}$, $e_{Q^*} \leq e^{\max}$. Moreover, it is straightforward to show that 
$l(e_Q)$ for both entropy and Gini index impurities are decreasing functions. Thus,  $I_{Q^*} \geq l(e_{Q^*}) \geq l(e^{\max})$. Therefore,
\begin{equation}
 \frac{I_{Q_{e^{\max}}}}{I_{Q^*}} \le \frac{u(e^{\max})}{\min_{e_Q}{l(e_Q)}} =  \frac{u(e^{\max})}{l(e^{\max})}=R(e^{\max}).   
\end{equation}

\begin{figure}[h]
  \centering
  $\begin{array}{cc}
  \hspace{-0.1 in} \includegraphics[width=1.8 in]{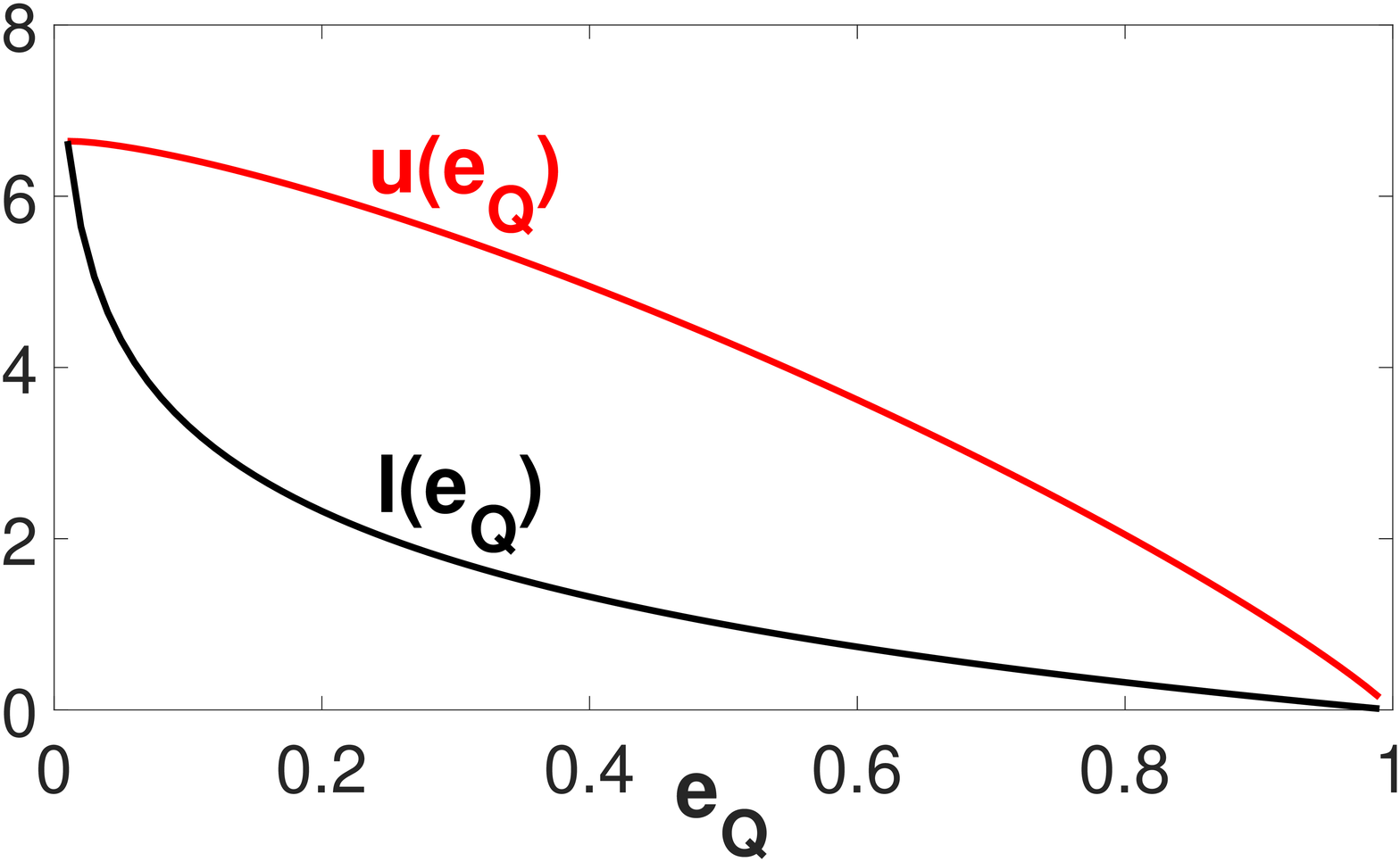} & \hspace{-0.1 in} \includegraphics[width=1.8 in]{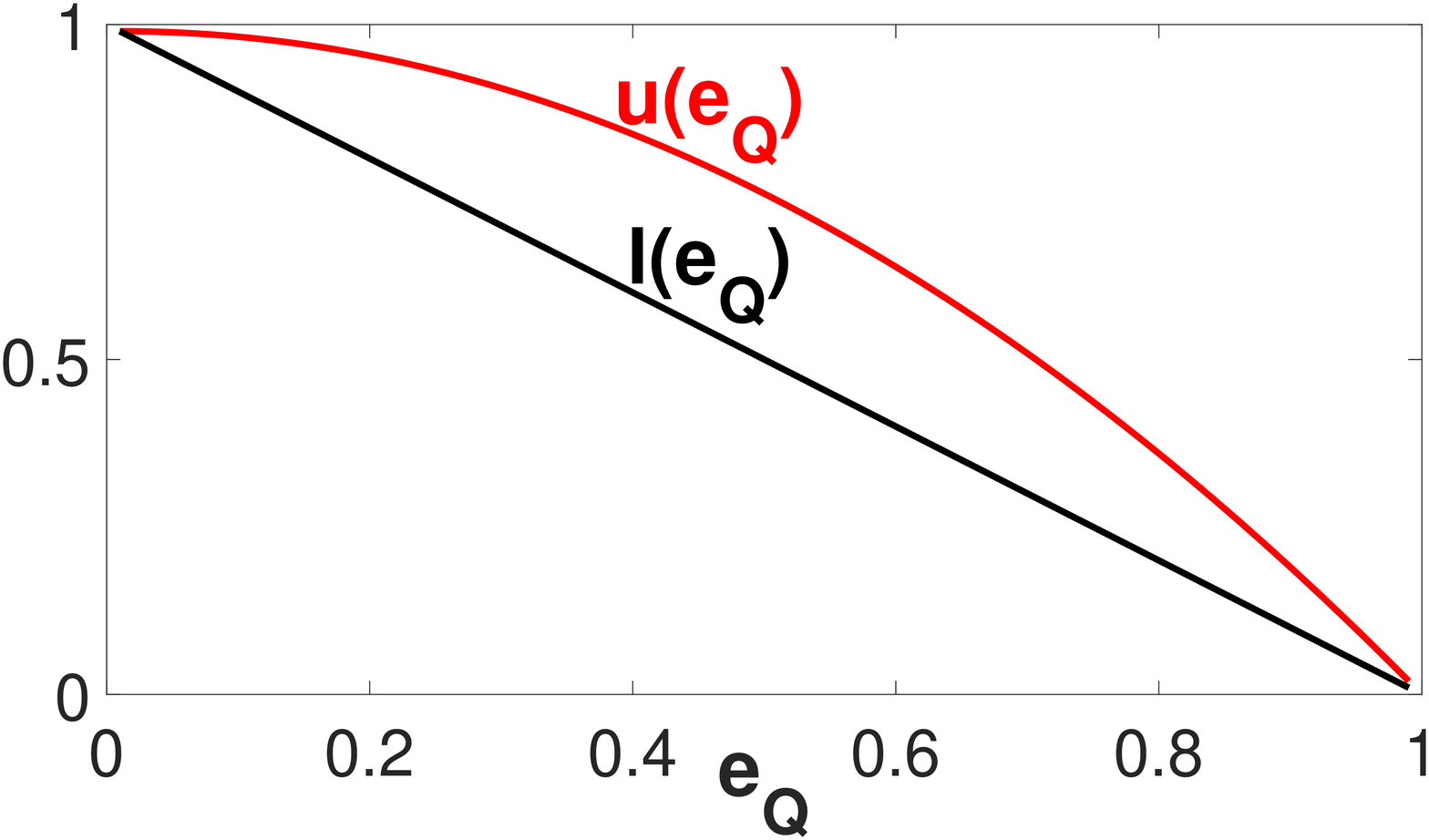} \\
  \hspace{-0.3 in} (a) & \hspace{-0.1 in} (b)
  \end{array}$
  \caption{The monotonic decreasing of $u(e_Q)$ and $l(e_Q)$ for (a) entropy impurity and (b) Gini index impurity using $e_Q \in (0.01,0.99)$ and $N=100$.  }\label{fig: 2-A}
\end{figure}

Thus, the impurity produced by Algorithm \ref{alg: finding emax} is guaranteed to be away from the true solution by at most a  factor of $R(e^{\max})$. Fig. \ref{fig: 2-A} shows $u(e_Q)$ and $l(e_Q)$ vs. $e_Q \in (0.01,0.99)$ using $N=100$ for both the entropy impurity and the Gini index impurity. As seen, $u(e_Q)$ and $l(e_Q)$ are monotonic decreasing functions for both entropy and Gini index impurities. Moreover, the upper bound and the lower bound are tight and equal when $e_Q=\dfrac{1}{N}$ or $e_Q=1$.
\end{proof}

The result in Theorem \ref{theorem: near opt} can be applied for any concave impurity function $f(x) = xl(x)$ with  $l(x)$ being a non-increasing function. Next, we show that $R(e^{\max})$-approximation is better than the approximation in \cite{cicalese2019new} for both the entropy impurity and the Gini index impurity.\\ 

\begin{theorem}
\textbf{}
\begin{itemize}
    \item For Gini index impurity,
    \begin{equation}
    \label{eq: remax for gini}
        R(e^{\max})=1+e^{\max}.
    \end{equation}
    \item For Entropy impurity,
    \begin{equation}
\label{eq: remax for entropy 2}
    R(e^{\max})=\dfrac{H(e^{\max})+(1-e^{\max})\log(N-1)}{-\log(e^{\max})}.
\end{equation}
\end{itemize}

\end{theorem}
\begin{proof}
The proof follows (\ref{eq: remax definition}) by using the upper bound and the lower bound in Theorem \ref{theorem: 1} and Theorem \ref{theorem: 2}, respectively. The detail of proof can be viewed in Appendix \ref{apd: proof of theorem 8} and \ref{apd: proof of theorem 9}.
\end{proof}

\begin{theorem}
\label{theorem: 8}
Algorithm \ref{alg: finding emax} provides a 2-approximation for Gini index impurity.
\end{theorem}

\begin{proof}
Please see Appendix \ref{apd: proof of theorem 8}.
\end{proof}

\begin{remark}
Algorithm 1 provides a 2-approximation for Gini index impurity while the algorithm in \cite{cicalese2019new} provides a 3-approximation in the worst case.
\end{remark}

\begin{theorem}
\label{theorem: 9}
The entropy impurity approximation provided by Algorithm \ref{alg: finding emax} is $R(e^{\max})$ and $R(e^{\max}) < \log^2 N $ if:
\begin{eqnarray}
\label{eq: condition for better approximation entropy impurity}
   N \geq N^{\min}=2^{S(e^{\max})},
\end{eqnarray}
where:
\begin{align*}
    S(e^{\max})=&\dfrac{1-e^{\max}}{-2\log (e^{\max})} \\&+ 
    \dfrac{\sqrt{4 H(e^{\max})(-\log(e^{\max}))+(1-e^{\max})^2}}{-2 \log(e^{\max})},
\end{align*}and $H(x)=-(x\log{x}+(1-x)\log(1-x))$ is the binary entropy of $x$. 
\end{theorem}

\begin{proof}
 Please see the Appendix \ref{apd: proof of theorem 9} for the details of proof. 
\end{proof}

\begin{remark}
If the condition in Theorem \ref{theorem: 9} is satisfied and $K \geq N$,  then $R(e^{\max}) < \log^2 N =\log^2(\min \{K,N\})$.  Therefore, Theorem \ref{theorem: 9} provides a sufficient condition where the approximation produced by Algorithm \ref{alg: finding emax} is better than that produced by the algorithm in \cite{cicalese2019new}. In reality, $R(e^{\max})$ is smaller than $\log^2 N$ for a wider range of $N$.
\end{remark}

\begin{figure}
\begin{center}
$\begin{array}{cc}
\hspace{-0.25 in} \includegraphics[scale=0.1]{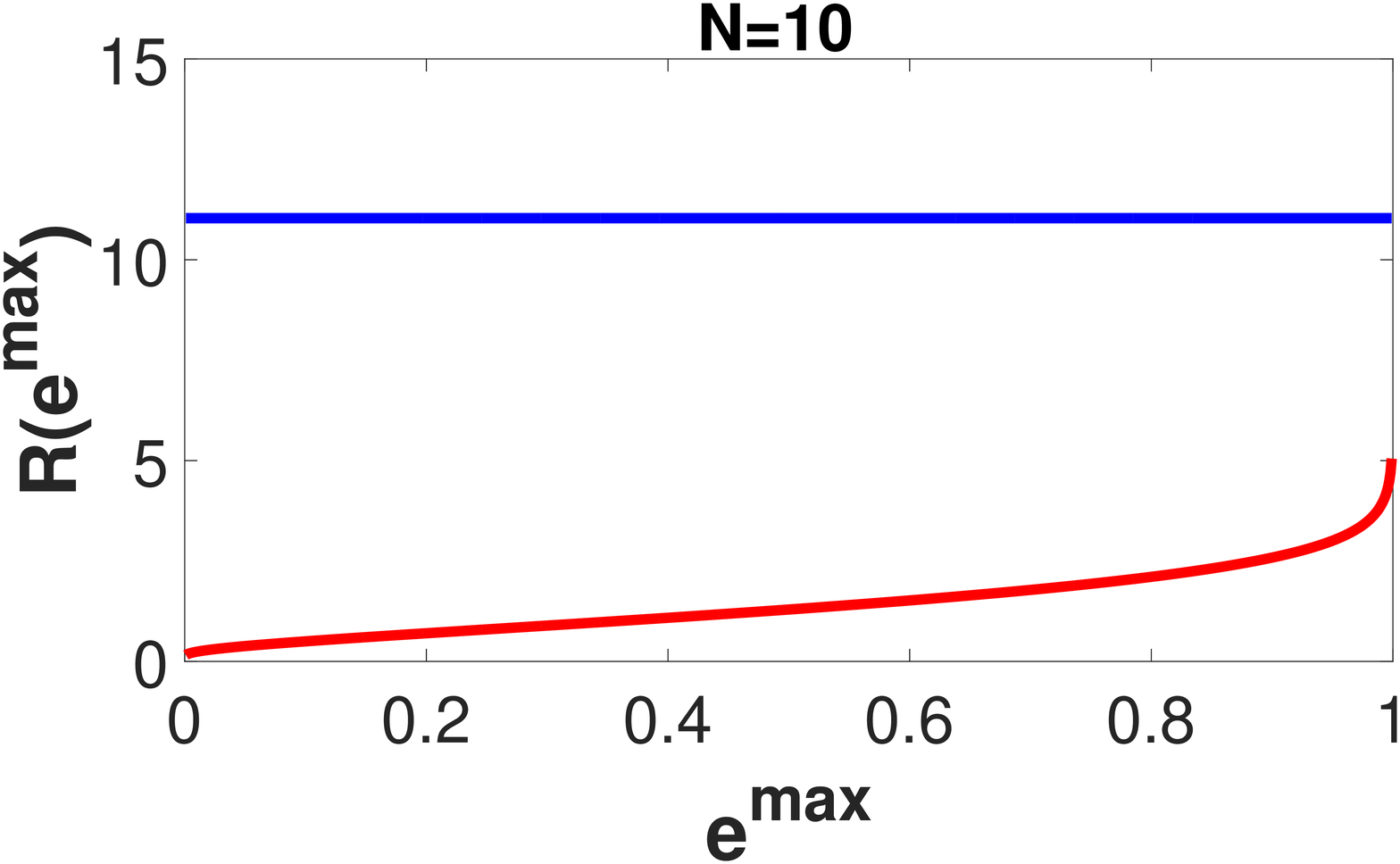} & 
\hspace{-0.3 in} \includegraphics[scale=0.1]{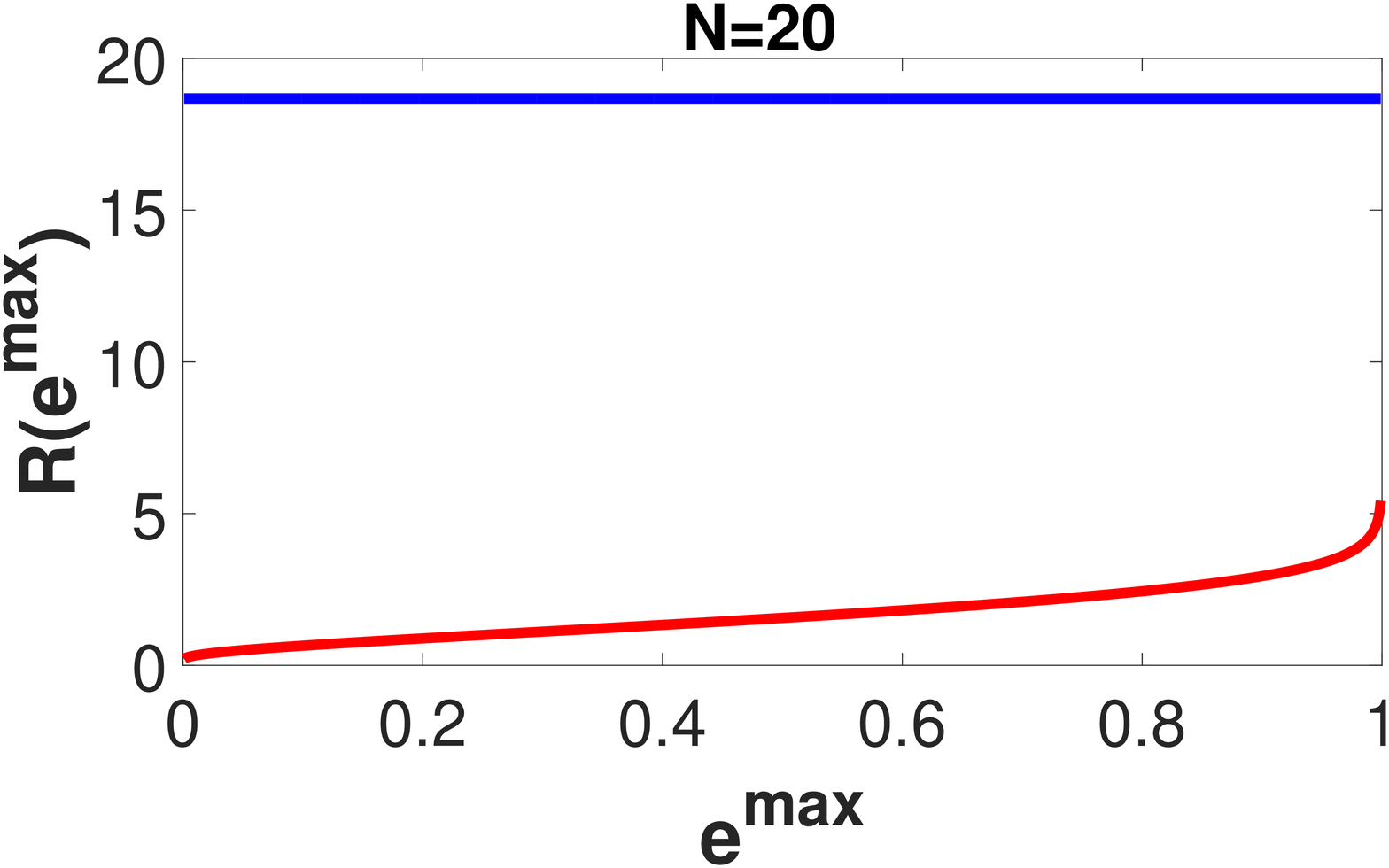}\\
\hspace{-0.15 in} (a) & \hspace{-0.2 in}(b) \\
\hspace{-0.25 in}\includegraphics[scale=0.1]{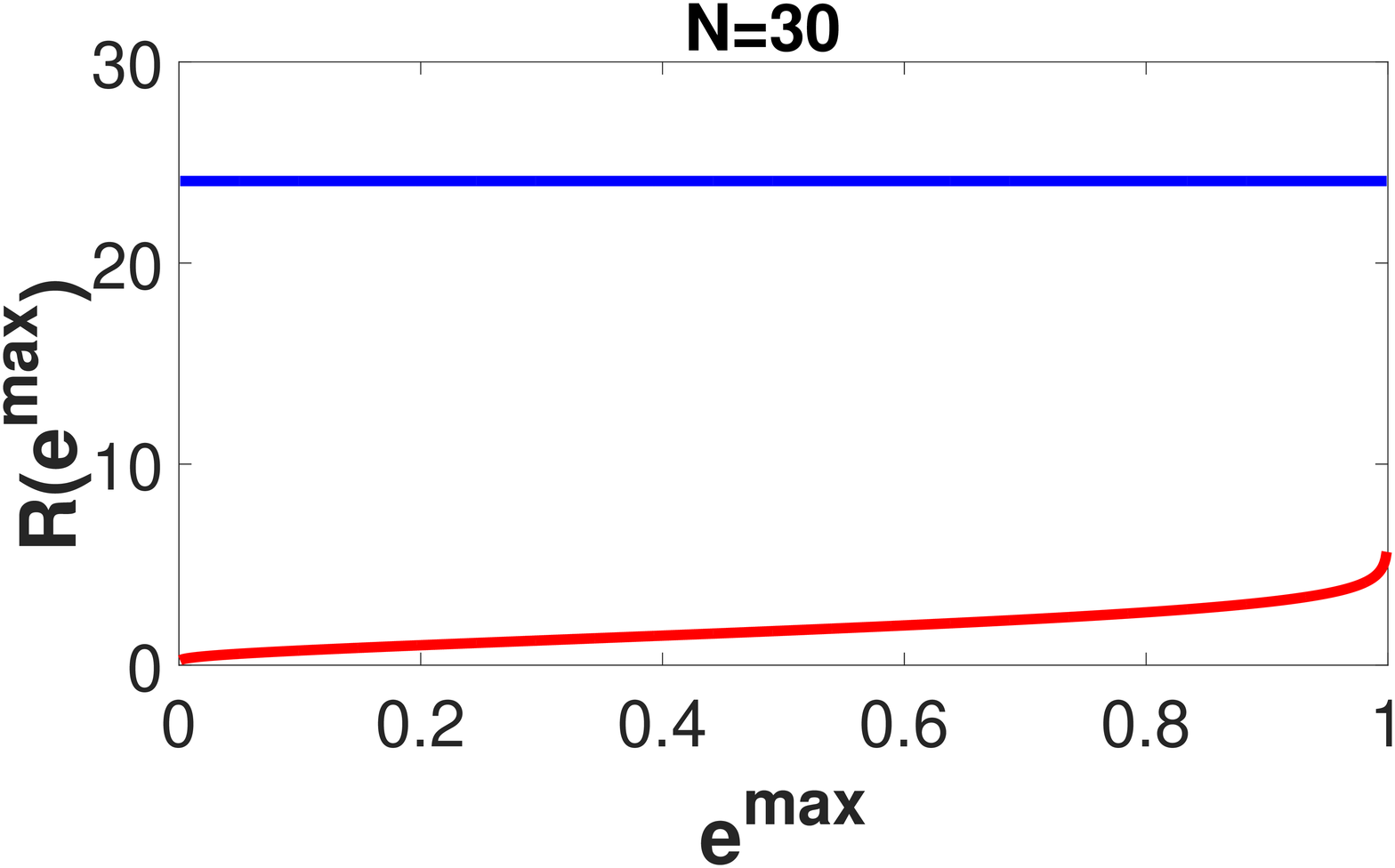} &
\hspace{-0.3 in}\includegraphics[scale=0.1]{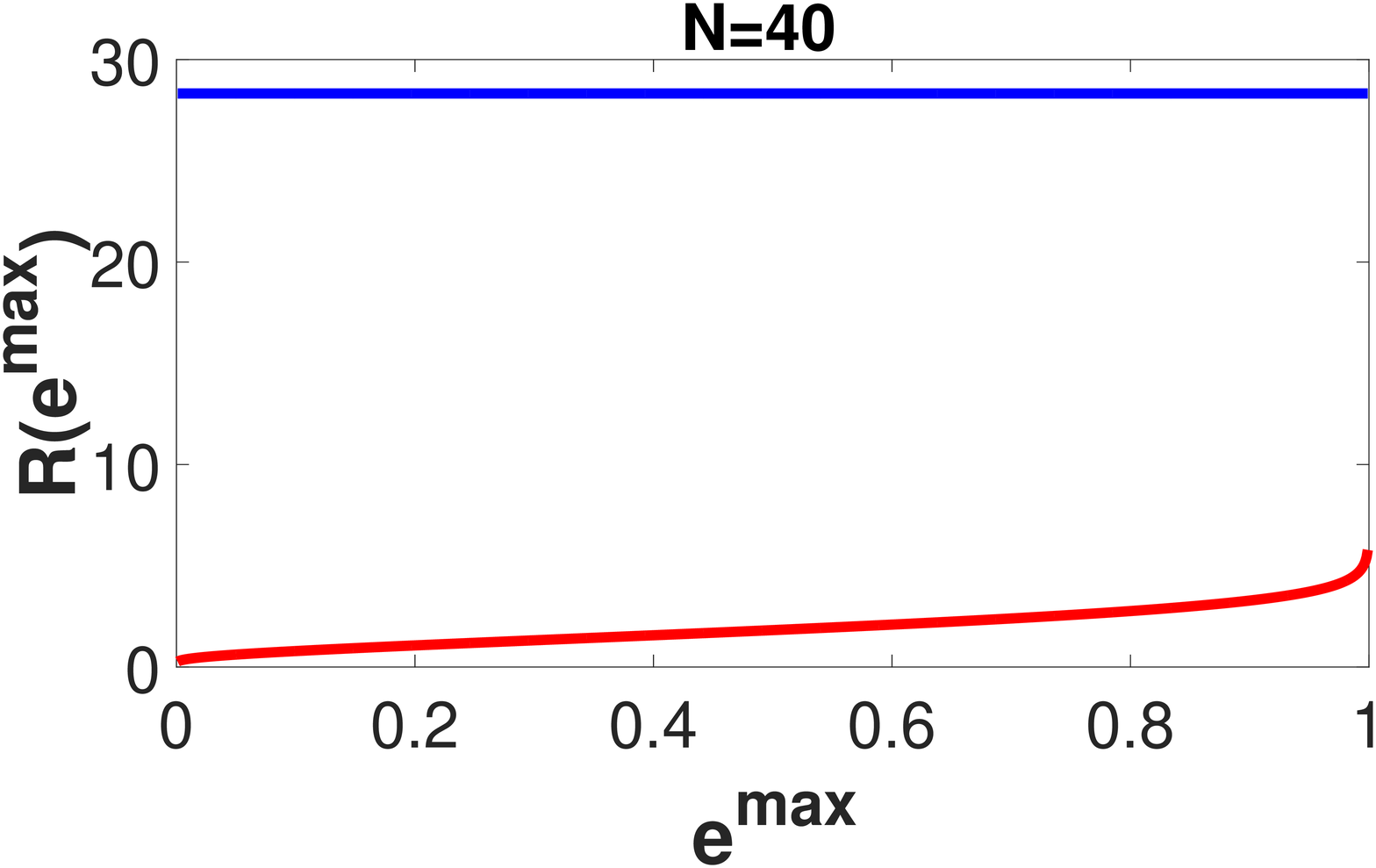} \\
\hspace{-0.15 in} (c) & \hspace{-0.2 in}(d)
\end{array}$
\end{center}
\caption{$R(e^{\max})$-approximation for  entropy impurity using (a) $N=10$; (b) $N=20$;  (c) $N=30$; (d) $N=40$.  Our approximation $R(e^{\max})$ are the red curves while the approximations of the algorithm in \cite{cicalese2019new} ($\log^2 N $, assume $K \geq N$) are the blue curves.}
\label{fig: 3}
\end{figure}

Fig. \ref{fig: 3} shows the performance bound of the proposed algorithm vs. the state of the art in \cite{cicalese2019new}.
$R(e^{\max})$ vs. $e^{\max} \in (0.01,0.99)$ for $N=10$, $N=20$, $N=30$, and $N=40$ are plotted in red while the approximations in \cite{cicalese2019new} ($\log^2 N $ when $K \geq N$) are plotted in blue. As seen,  the red curves are always below the blue curves. Moreover, the gaps between our approximation and that of \cite{cicalese2019new} are proportional to the size of $N$. That said, for large values of $N$, our approximation is progressively better than that of \cite{cicalese2019new}. 

We also note that $S(e^{\max})$ is a monotonic increasing function as shown in Fig. \ref{fig: 2}-(a).
Thus, if $e^{\max}$ increases, then $N^{\min}$ increases. For example, if $e^{\max}=0.5$ then (\ref{eq: condition for better approximation entropy impurity}) holds for $N^{\min}=2.42$, if $e^{\max}=0.8$, (\ref{eq: condition for better approximation entropy impurity}) holds for $N^{\min} = 3.58$, if $e^{\max}=0.9$, (\ref{eq: condition for better approximation entropy impurity}) holds for $N^{\min} = 4.34$, if $e^{\max}=0.999$, (\ref{eq: condition for better approximation entropy impurity}) holds for $N^{\min}= 9.06$. 
Fig. \ref{fig: 2}-(b) illustrates the relationship between $e^{\max}$ and $N^{\min}$. We also note that for the real datasets, $e^{\max}$ is typically small which results in a small value of $N^{\min}$. For example, $e^{\max}$ are $0.2420$ and $0.2185$, and $N^{\min}$ are $2.15$ and $2.09$ for 20NEWS dataset and RCV1 dataset, respectively. Please see Table I and Table II in Section \ref{sec: simulation} for the values of $e^{\max}$. 

\begin{figure}
  \centering
  $\begin{array}{cc}
  \hspace{-0.2 in} \includegraphics[width=1.9 in]{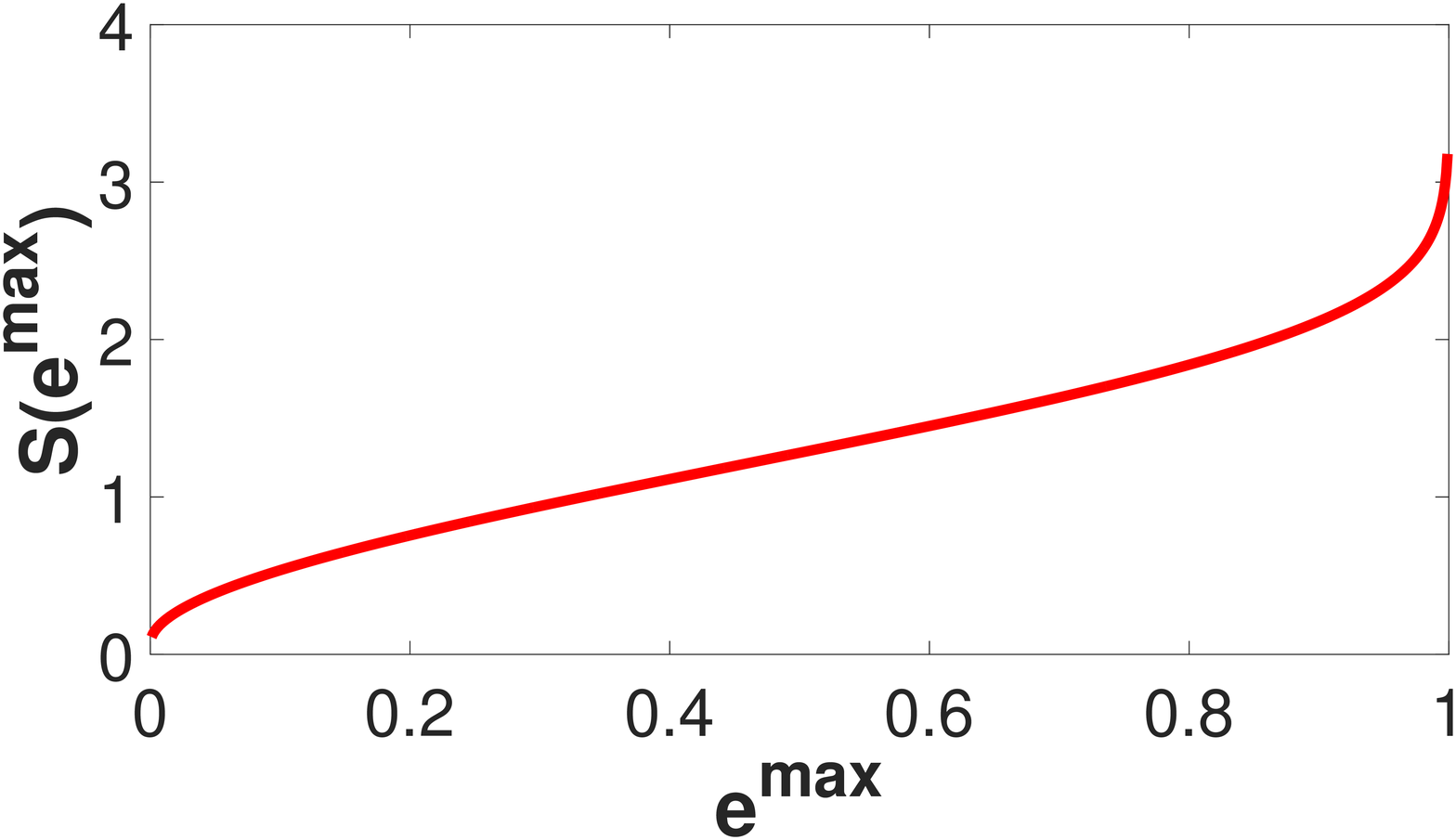} & \hspace{-0.25 in} \includegraphics[width=1.9 in]{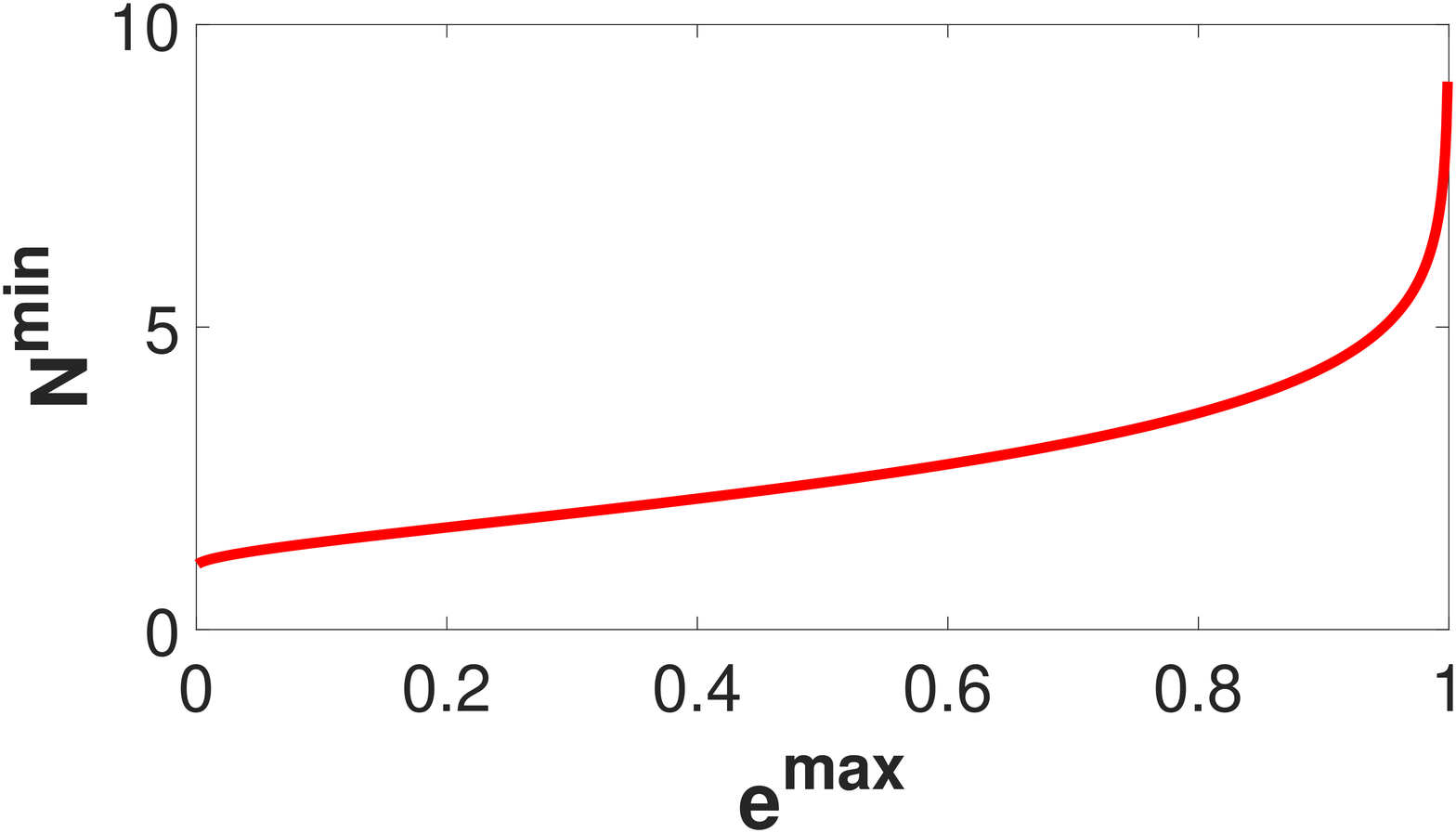} \\
  \hspace{-0.2 in} (a) & \hspace{-0.25 in} (b)
  \end{array}$
  \caption{(a) $S(e^{\max})$ as a function of $e^{\max}$; (b) $N^{\min}=2^{S(e^{\max})}$ as a function of $e^{\max}$.  }\label{fig: 2}
\end{figure}

\begin{remark} For $K \geq N$, the approximation guaranteed by Algorithm \ref{alg: finding emax} for entropy impurity is better than that of the state-of-art approximation in \cite{cicalese2019new} for most the value of $N$.   
On the other hand, when $K <N$, it is possible that the approximation in \cite{cicalese2019new} provides a better bound than our approximation i.e., $\log^2 (\min \{K,N\}) < R(e^{\max})$ due to $\min \{K,N\}=K$ completely depends on $K$ while $R(e^{\max})$ in (\ref{eq: remax for entropy 2}) depends on both $e^{\max}$ and $N$. Therefore, a smaller value of $K$, a higher chance that the algorithm in \cite{cicalese2019new} provides a better approximation. For example, consider the 20NEWS dataset having  $e^{\max}=0.2420$, using $N=20$, the approximation in \cite{cicalese2019new} is better than our approximation if $K \leq 2.64$. Similarly, consider the RCV1 dataset having $e^{\max}=0.2185$, using $N=103$, the approximation in \cite{cicalese2019new} is better than our approximation if $K \leq 5.97$. To that end, our algorithm still provides a better approximation for a wide range of $K$ even if $K < N$. For example, our bound is better $\forall K \geq 2.64$ using 20NEWS dataset, and $\forall K \geq 5.97$ using RCV1 dataset regardless of $N$. Please see detail of these datasets in Sec. \ref{sec: simulation}. 
\end{remark}

\section{Enhanced Algorithms}
\label{sec: greedy algs}

In the previous section, we show that the proposed Algorithm \ref{alg: finding emax} is near-optimal.   However, there are some main drawbacks that limit the applications of Algorithm \ref{alg: finding emax}. Particularly,
\begin{itemize}
    \item When $K>N$, Algorithm \ref{alg: finding emax} produces the optimal partitions containing exactly $N$ partitions due to $\mathcal{V} = \mathcal{V}_N$. Therefore, even though that Algorithm \ref{alg: finding emax} still achieves the theoretical bounds, it produces $(K-N)$ empty partitions which is less useful. 
    
    \item When $K<N$, the time complexity of Algorithm \ref{alg: finding emax} in the worst case is exponential in $N$ which is less desirable when the dimension of the data is large. 
    
    \item The partitions produced by Algorithm \ref{alg: finding emax} might not be locally optimal partitions. The necessary condition for optimal partitions can be viewed in Theorem 1 of \cite{coppersmith1999partitioning}.
\end{itemize}
To resolve these problems, we propose several modifications of Algorithm \ref{alg: finding emax} which results in two linear-time complexity algorithms. These algorithms are based on greedy-splitting (Algorithm \ref{alg: greedy splitting}) and greedy-merge (Algorithm \ref{alg: greedy merge}) of the partitions produced by Algorithm \ref{alg: finding emax}. 
\vspace{-10 pt}
\subsection{Handling the case $K > N$: greedy-splitting algorithm}
\label{apd: handingle K > N}
\vspace{-5 pt}
\begin{footnotesize}
\begin{algorithm}
\caption{ Greedy-splitting algorithm for $K > N$.}
\label{alg: greedy splitting}
\begin{algorithmic}[1]
\State{\textbf{Input}: Dataset $Y=\{y_1,\dots,y_M\}$ and  $p(x_i,y_j)$}.

\State{\textbf{Output}: Partition $Z=\{z_1,z_2,\dots,z_K\}$}.  

\State{\hspace{0.2 in} \textbf{Step 1}: Run Algorithm \ref{alg: finding emax} to obtain $N$-partition $z_1,z_2,\dots,z_N$.}

\State{\hspace{0.2 in} \textbf{Step 2:} Greedy splitting.}

\State{\hspace{0.4 in} \textbf{$t=1$}}

\State{\hspace{0.4 in} \textbf{While: $t \leq K-N$}}

\State{\hspace{0.6 in} Finding the largest impurity.}
\begin{equation}
\label{eq: finding largest impurity}
    i^*=\max_i F(\mathbf{p}_{\mathbf{x},z_i}).
\end{equation}

\State{\hspace{0.6 in} Splitting based on the largest attribution.}
\begin{equation}
    j^*=\max_j p(x_j|z_{i^*}).
\end{equation}
\State{\hspace{0.6 in} For $y_q \in z_{i^*}$.}

\State{\hspace{0.8 in} If $p(x_{j^*}|y_q) \leq p(x_{j^*}|z_{i^*})$.}

     $$Q(y_q) \rightarrow z_{i^*}.$$
     
 \State{\hspace{0.8 in} Else $p(x_{j^*}|y_q) > p(x_{j^*}|z_{i^*})$.}  
 
    $$Q(y_q) \rightarrow z_{N+t}.$$

\State{\hspace{0.6 in} $t=t+1$.}

\State{\hspace{0.4 in} \textbf{End While}.}

\State{\hspace{0.2 in} \textbf{Step 3}: Return $K$ partitions.}

\State{\textbf{Return:} Return $K$ partitions.}
\end{algorithmic}
\end{algorithm}
\end{footnotesize}

To resolve the problem of $(K-N)$ empty partitions when $K>N$, we propose a so-called greedy-splitting algorithm (Algorithm \ref{alg: greedy splitting}). The first step of greedy-splitting algorithm is using Algorithm \ref{alg: finding emax} to generate $N$ non-empty partitions.  Next, by greedy splitting, one can  generate more $(K-N)$ non-empty partitions to obtain total $K$ partitions. As will be shown later, Algorithm \ref{alg: greedy splitting} runs in  $O(KNM)$ and achieves all the theoretical bounds of Algorithm \ref{alg: finding emax}.

\textbf{Algorithm:} The first step of greedy-splitting algorithm is finding the partition that has the largest impurity (line 7, Algorithm \ref{alg: greedy splitting}). Next, this partition is separated based on the largest attribution.  For example, if the largest impurity partition is $z_{i^*}$, and recall that: $$j^*=\max_{1 \leq j \leq N} p(x_j|z_{i^*}).$$ 
Then  $\forall y_q \in z_{i^*}$, $p(x_{j^*}|y_q)$ is the largest attribution of $y_q$. Using $p(x_{j^*}|z_{i^*})$ as a threshold, by comparing $p(x_{j^*}|y_q)$ to $p(x_{j^*}|z_{i^*})$, $ \forall y_q \in z_{i^*}$, $y_q$ is assigned into two new partitions (line 10 and 11, Algorithm \ref{alg: greedy splitting}).  The process repeats until $K$ partitions are generated.  Although the splitting based on $p(x_{j^*}|z_{i^*})$ as a threshold is a heuristic method, it guarantees a better impurity than that provided by Algorithm \ref{alg: finding emax} as will be shown later. The pseudo-code of our splitting procedure is presented in Algorithm \ref{alg: greedy splitting}.

\textbf{Proof of better approximation:} From Proposition 1 in \cite{coppersmith1999partitioning}, if $z_k= z_k^1 \cup z_k^2$ and $z_k^1 \cap z_k^2 = \emptyset$, then the impurity in $z_k$ is at least as large as the total impurity in $z_k^1$ and $z_k^2$. In other words, the impurity of a set always decreases after splitting.
Therefore, by splitting, Algorithm \ref{alg: greedy splitting} produces a new partition having the impurity is monotonically decreased  over each splitting step. Finally, the impurity of $K$ partitions produced by greedy-splitting algorithm is less than or at least equal the impurity provided by Algorithm \ref{alg: finding emax}. Thus, the partitions induced by greedy-splitting algorithm must satisfy our theoretical bounds in Sec. \ref{sec: proof of near optimal}.  

\textbf{Time complexity of Algorithm \ref{alg: greedy splitting}:} The time complexity of Step 1 and Step 2  of Algorithm \ref{alg: greedy splitting} are $NM$ and $(K-N)NM$, respectively. Therefore, the time complexity of Algorithm \ref{alg: greedy splitting} is $O(KNM)$.

\vspace{-10 pt}
\subsection{Handing the case $K<N$: greedy-merge algorithm}
\label{apd: handling K < N}
\vspace{-5 pt}
\begin{footnotesize}
\begin{algorithm}
\caption{ Greedy-merge algorithm.}
\label{alg: greedy merge}
\begin{algorithmic}[1]
\State{\textbf{Input}: Dataset $Y=\{y_1,\dots,y_M\}$ and  $p(x_i,y_j)$}.

\State{\textbf{Output}: Partition $Z=\{z_1,z_2,\dots,z_K\}$}. 

\State{\hspace{0.2 in} \textbf{Step 1}: Run Algorithm \ref{alg: finding emax} to obtain $N$-partition $z_1,z_2,\dots,z_N$.}

\State{\hspace{0.2 in} \textbf{Step 2}: Greedy merge.}

\State{\hspace{0.4 in} $t=0$}

\State{\hspace{0.4 in} \textbf{While:} $t \leq N-K$}

\State{\hspace{0.6 in} \textbf{For}: $i=1,2,\dots,K-t-1$}

\State{\hspace{0.8 in} \textbf{For}: $j=i+1,i+2,\dots,K-t$}

\State{\hspace{0.8 in} Merge $z_i$, $z_j$ to $z_{ij}$ and compute:}
\begin{equation}
   \mathbf{p}_{\mathbf{x},z_{ij}}=\mathbf{p}_{\mathbf{x},z_{i}}+ \mathbf{p}_{\mathbf{x},z_{j}},
\end{equation}
\begin{equation}
   F(\mathbf{p}_{\mathbf{x},z_{ij}})=\sum_{k=1}^{N}p(x_k,z_{ij})\sum_{k=1}^{N} f \Big( \dfrac{p(x_k,z_{ij})}{\sum_{i=1}^{N}p(x_k,z_{ij})}\Big),
\end{equation}
\State{\hspace{0.8 in}}
\begin{equation}
\label{eq: delta}
     \Delta_{ij}=F(\mathbf{p}_{\mathbf{x},z_{ij}})- F(\mathbf{p}_{\mathbf{x},z_{i}})-F(\mathbf{p}_{\mathbf{x},z_{j}}). 
\end{equation}
\State{\hspace{0.8 in}\textbf{End For }}
\State{\hspace{0.6 in}\textbf{End For }}

\State{\hspace{0.6 in}Return the best merge that minimizes the impurity loss $\Delta_{ij}$:}
\begin{equation}
\label{eq: delta merge}
    i^*,j^*=\min_{i,j} \Delta_{ij}.
\end{equation}

\State{\hspace{0.4 in}\textbf{End While }}
\State{\hspace{0.2 in} \textbf{Step 3}: Return $K$ partitions.}

\State{\textbf{Return:} Return $K$ partitions.}
\end{algorithmic}
\end{algorithm}
\end{footnotesize}

To resolve the high time complexity of Algorithm \ref{alg: finding emax} when $K<N$, we propose a so-called greedy-merge algorithm (Algorithm \ref{alg: greedy merge}) to reduce the time complexity. Particularly, we first use the Algorithm \ref{alg: finding emax} for $K=N$ to obtain $N$-partition. Next, we perform $N-K$ times of the greedy-merge, each time the algorithm merges two partitions into one single partition that minimizes the impurity loss until exact $K$ partitions are obtained ($K<N$). As will be shown later, the time complexity of this greedy-merge algorithm is $O((N-K)N^2+NM)$ that is linear in $M$ and polynomial in $N$. Although the greedy-merge algorithm does not satisfy the theoretical bounds, it performance  is comparable to the results provided by the proposed algorithm in \cite{cicalese2019new}. Please see the numerical results in Sec. \ref{sec: simulation}.

\textbf{Algorithm: } As discussed earlier, greedy-merge algorithm first uses Algorithm \ref{alg: finding emax} to generate $N$-partition (line 3, Algorithm \ref{alg: greedy merge}). Next, greedy-merge algorithm performs $(N-K)$ greedy merges, each time the algorithm merges two partitions into one single partition that minimizes the impurity loss $\Delta$ (line 10 and 13, Algorithm \ref{alg: greedy merge}) until exact $K$ partitions are obtained. The pseudo-code of our greedy-merge algorithm is provided in Algorithm \ref{alg: greedy merge}.

\textbf{Time complexity of Algorithm \ref{alg: greedy merge}:} The  time complexities of Step 1 and 2 in Algorithm \ref{alg: greedy merge} are $NM$ and $(N-K)N^2$, respectively. Thus, the time complexities of Algorithm \ref{alg: greedy merge} is $O(NM+(N-K)N^2)$. 
\vspace{-10 pt}
\subsection{Reaching to the local optimal solutions}
In \cite{coppersmith1999partitioning}, the authors proposed a necessary condition for which the partition is optimal. As the result, an iterative based $k$-means algorithm with a suitable distance can be used to find locally optimal partitions  \cite{coppersmith1999partitioning}, \cite{chou1991optimal}, \cite{cicalese2019new}, \cite{zhang2016low}, \cite{banerjee2005clustering}. On the other hand, although our proposed algorithms can approximate globally optimal solution well, there is no guarantee that the produced partitions are optimal i.e.,  the produced partitions might not satisfy the optimality condition in Theorem 1 of \cite{coppersmith1999partitioning}. Therefore, one can always perform iterative algorithms over the  partitions produced by Algorithm \ref{alg: finding emax}, Algorithm \ref{alg: greedy splitting} and Algorithm \ref{alg: greedy merge}  to obtain locally optimal solutions. This optional step will improve the quality of our proposed algorithms at the expense of increased time complexity $O(TKNM)$ where $T$ is the number of iterations.

\section{Numerical results}
\label{sec: simulation}
We used two datasets: 20NEWS and RCV1.  These are widely used for evaluating text classification methods \cite{cicalese2019new}. Existing algorithms \cite{coppersmith1999partitioning}, \cite{chou1991optimal}, \cite{zhang2016low}, \cite{banerjee2005clustering} can only find locally optimal solutions.  To approximate a globally optimal solution, many iterative algorithms use multiple random starting points and select the best solution. To that end, we compare the impurity provided by Algorithm \ref{alg: greedy splitting} when $K \geq N$ and Algorithm \ref{alg: greedy merge} when $K<N$ with the impurity produced by running the iterative algorithms 100 times from 100 randomly starting points. The details of these iterative algorithms can be viewed in \cite{coppersmith1999partitioning}, \cite{chou1991optimal}, \cite{cicalese2019new}, \cite{zhang2016low}, \cite{banerjee2005clustering}. 
Although these iterative algorithms do not guarantee to find a globally optimal solution, their performances were shown in \cite{banerjee2005clustering} to outperform the well-known Agglomerative Clustering method   in \cite{slonim2001power}, \cite{slonim2000agglomerative}. 

\begin{figure}
  \centering
  $\begin{array}{c}
    \includegraphics[width=.45\textwidth]{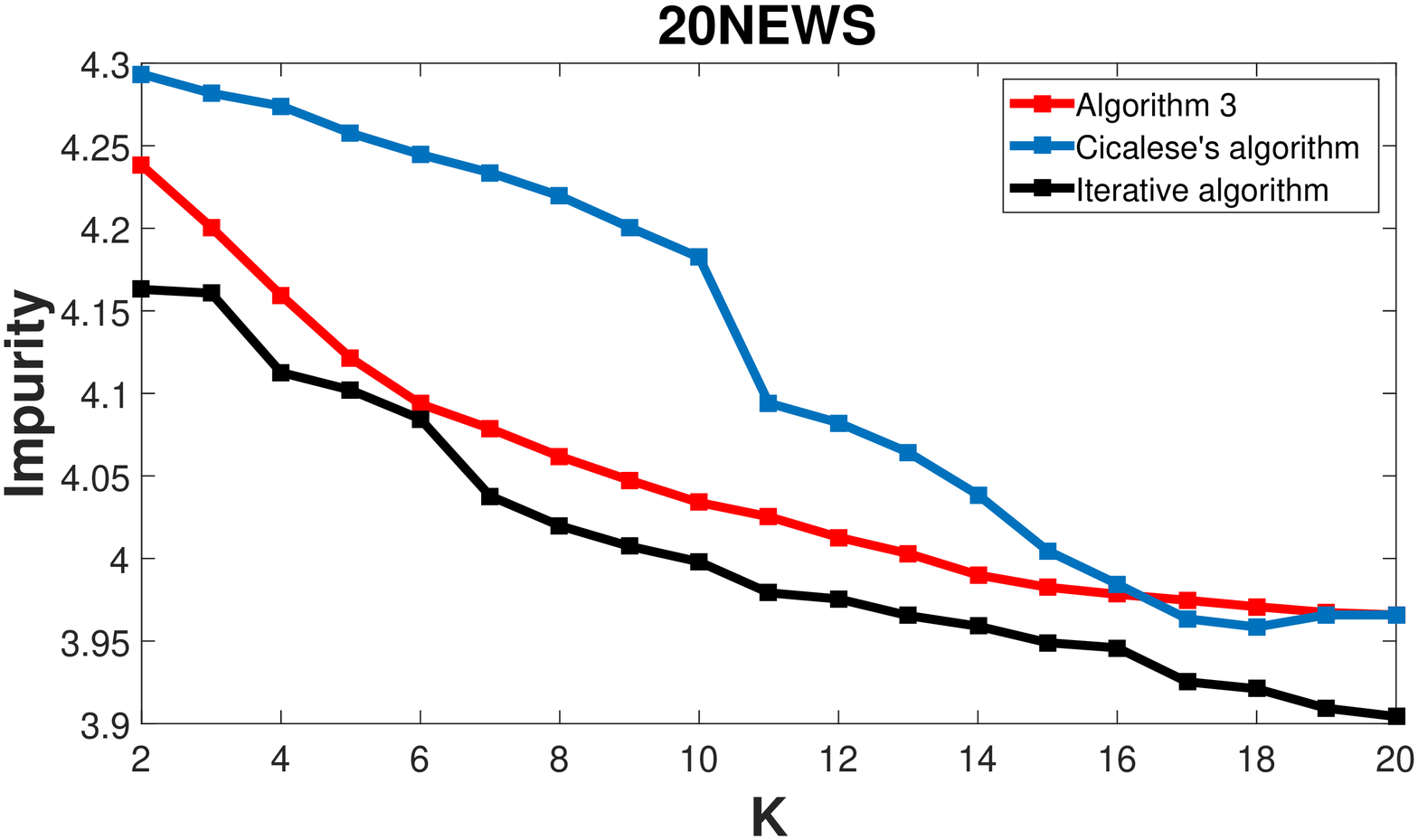} \\
  (a)\\
  \includegraphics[width=.45\textwidth]{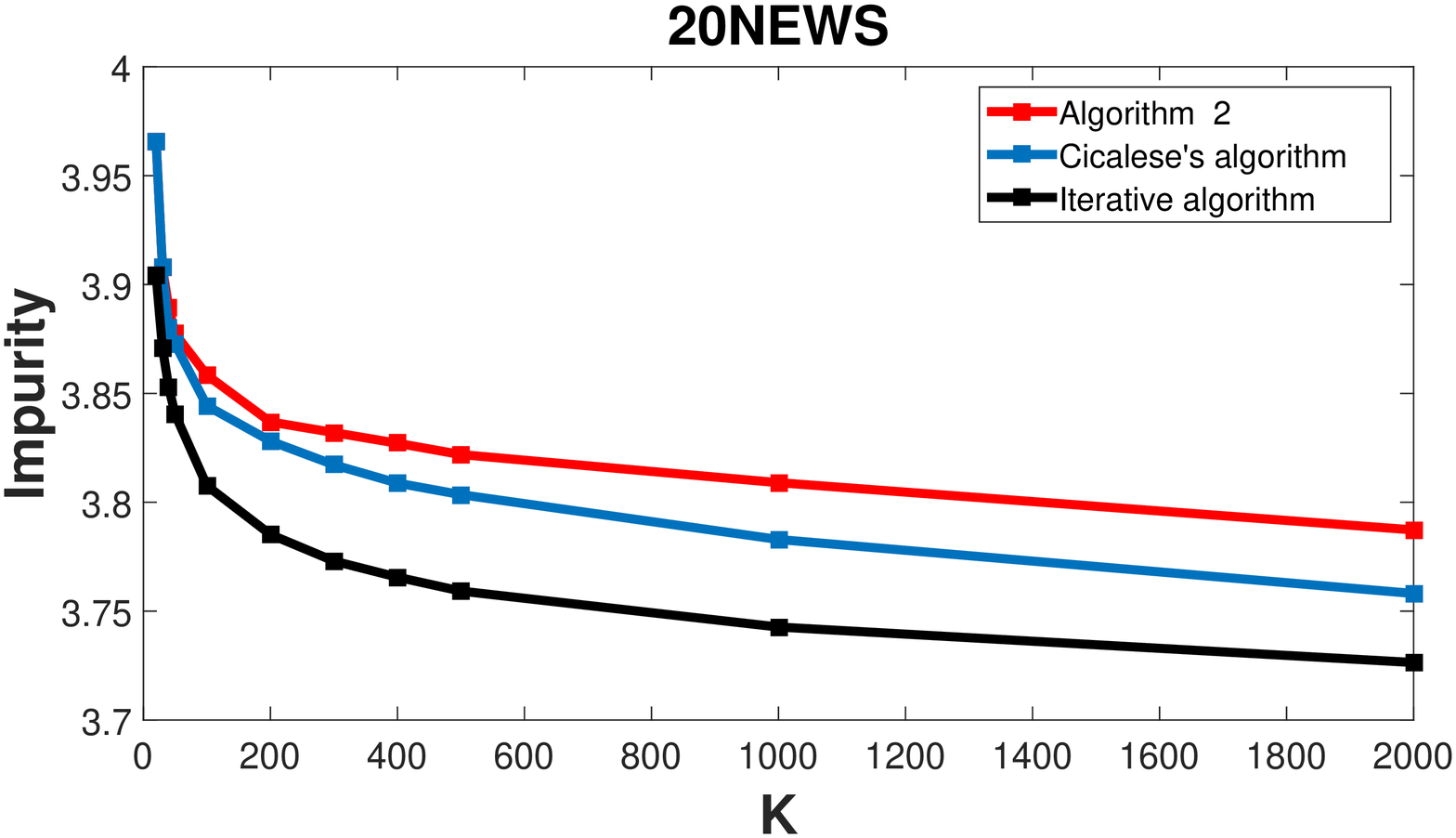} \\
  (b)
  \end{array}$
  \vspace{-0.1 in}
  \caption{Simulation results using 20NEWS dataset: (a) Algorithm 3 vs. the proposed algorithm in \cite{cicalese2019new} when $K<N$; (b) Algorithm 2 vs. the proposed algorithm in \cite{cicalese2019new} when $K \geq N$.  }\label{fig: simulation}
\end{figure}
\begin{figure}
  \centering
  $\begin{array}{c}
   \includegraphics[width=.45\textwidth]{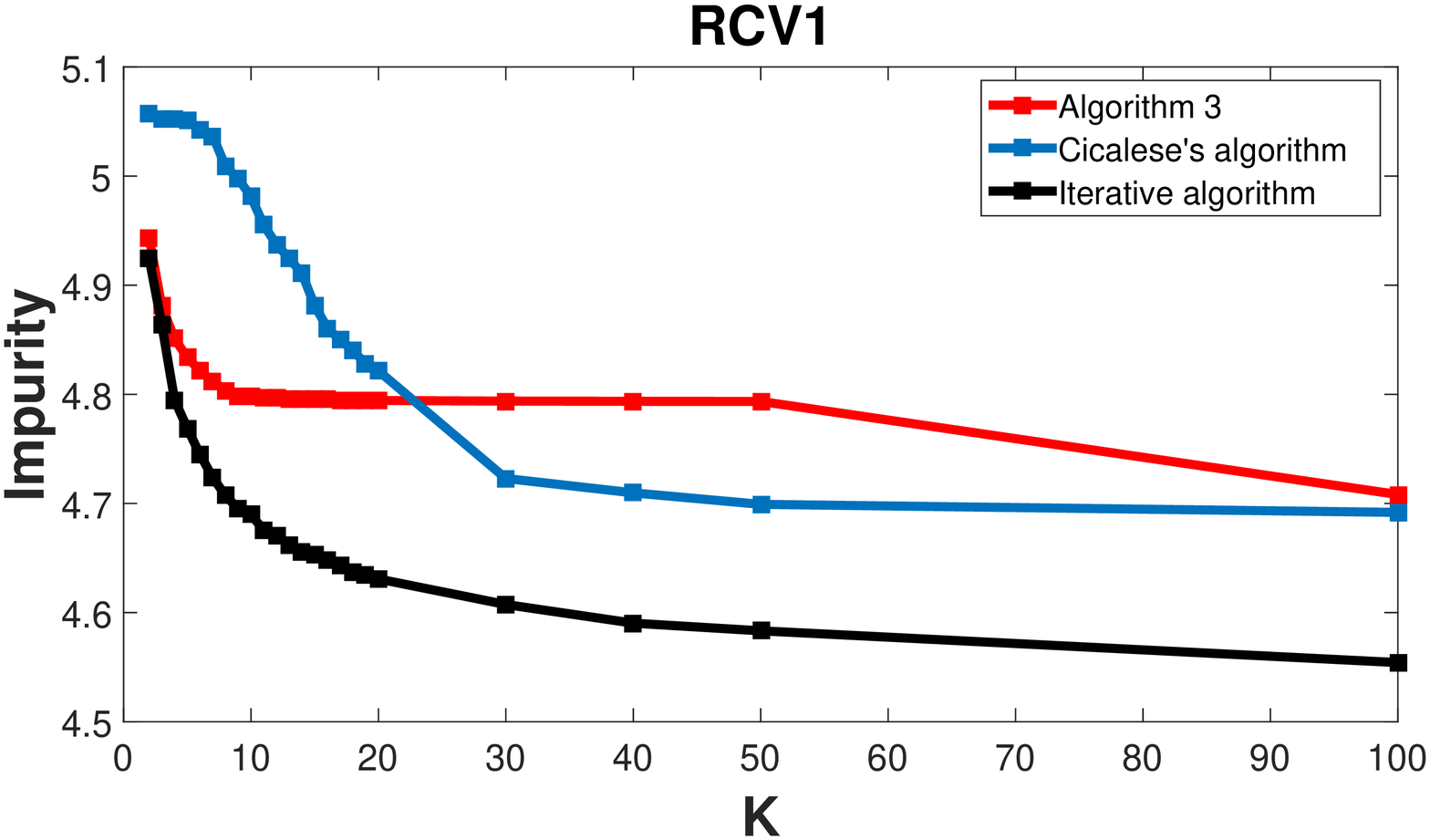} \\
  (a)\\ 
 \includegraphics[width=.45\textwidth]{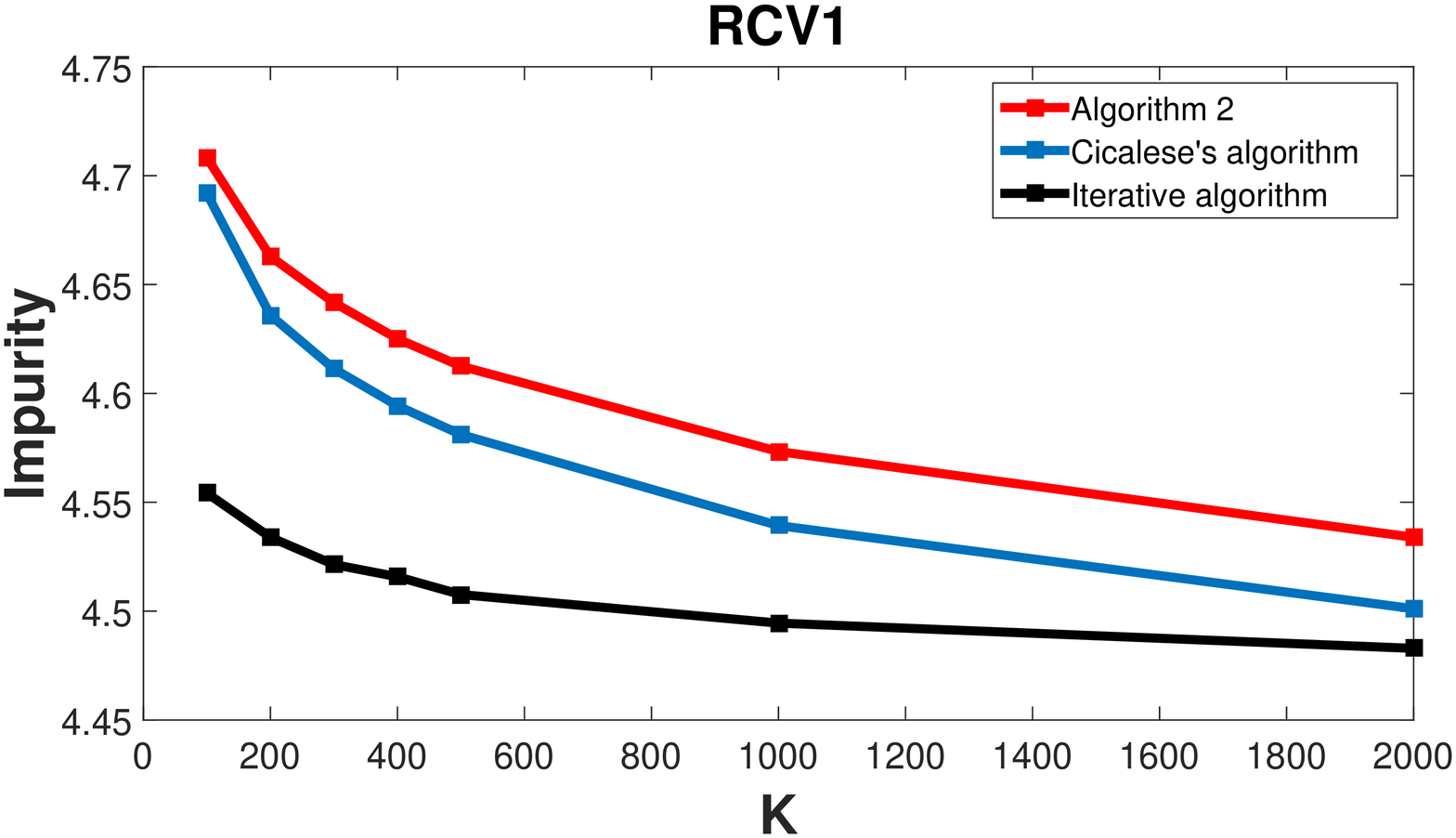} \\
   (b)
  \end{array}$
    \vspace{-0.1 in}
  \caption{Simulation results using RCV1 dataset: (a) Algorithm 3 vs. the proposed algorithm in \cite{cicalese2019new} when $K<N$; (b) Algorithm 2 vs. the proposed algorithm in \cite{cicalese2019new} when $K \geq N$.  }\label{fig: simulation 2}
\end{figure}

20NEWS dataset includes 18.846 documents evenly divided into 20 disjoint classes while RCV1 dataset includes 804,414 documents assigned to 103 different classes.  
Since both our algorithms and iterative algorithms use the joint distribution dataset, one wants to normalize the raw data in 20NEWS and RCV1 to a joint distribution $p(x_i,y_j)$, for example, by counting the number of times that a word $y_j$ appears in document $x_i$. For convenience, we utilize the normalized datasets in \cite{cicalese2019new}. 
After normalized, the dataset 20NEWS contains $M= 51840$ vectors of dimension $N=20$ while the dataset  RCV1 has $M=170946$ vectors of dimension $N=103$. Our code as well as the datasets are available at  \url{https://github.com/hoangle96/linear_clustering}.  

Next, we run the proposed algorithms (Algorithm \ref{alg: greedy splitting} and \ref{alg: greedy merge} corresponding to the case of $K \geq N$ and $K<N$, respectively), the iterative algorithm, and the ratio-greedy algorithm in \cite{cicalese2019new} for $K=2,3,4,5,\dots,2000$ using  both 20NEWS dataset and RCV1 dataset. 
Figs. \ref{fig: simulation} and \ref{fig: simulation 2} illustrate the impurity provided by our proposed Algorithm \ref{alg: greedy splitting} and \ref{alg: greedy merge}, the iterative algorithms, and the algorithm in \cite{cicalese2019new} for 20NEWS and RCV1 datasets. 
As seen, the impurities resulted from our proposed algorithms are very close to the impurity obtained from the iterative algorithms (assuming that the iterative algorithm obtains a globally optimal solution). Particularly, the impurities resulted by our proposed algorithms are at most $1.0181$ time larger than the impurity provided by running the iterative algorithms $100$ times on 20NEWS dataset and at most $1.0459$ time larger for the RCV1 dataset. In addition, the impurities provided by our algorithms (the red curves) are comparable to the impurities obtained by the proposed algorithm in \cite{cicalese2019new} (the blue curves) as illustrated in Fig. \ref{fig: simulation} and Fig. \ref{fig: simulation 2}. Particularly, Fig. \ref{fig: simulation} and Fig. \ref{fig: simulation 2} show out that our proposed algorithms outperform the proposed algorithm in \cite{cicalese2019new} if the number of partitions $K$ is small. When $K$ is large (for example $K \geq 19$ for 20NEWS and $K\geq 24$ for RCV1), the algorithm in \cite{cicalese2019new} provides lower impurities, but our proposed algorithm has theoretically lower time complexities than other algorithms.


 \section{Conclusion}
 \label{sec: conclusion}
In this paper, we propose a new algorithm with bounded guarantee splitting quality based on the maximum likelihood principle for minimizing a wide class of impurity function including entropy and Gini index.  Under certain conditions, we show that the proposed algorithm is better than the state-of-art algorithms in both terms of computational complexity and the quality of partitioned outputs. We also provide two heuristic algorithms that works well in practice.  In addition, our new upper and lower bounds generalize two well-known results in information theory and signal processing, namely the Fano's inequality and the Boyd-Chiang's upper bound of channel capacity.

\bibliographystyle{unsrt}
\bibliography{sample}

\appendix
\small

\subsection{Improvement of Algorithm in \cite{cicalese2019new}}
\label{apd: time complexity analysis}

In Section V \cite{cicalese2019new}, Cicalese \textit{et al.} proposed an algorithm that provably achieves near-optimal partition for entropy impurity. This algorithm has two steps: (1) performing a projection to transfer the multidimensional data back to  a 2-dimensional data, and (2) using dynamic programming to find the optimal partition in 2-dimensional data based on the idea in \cite{kurkoski2014quantization}.

Cicalese \textit{et al.} proved that the time comlexity of the algorithm in \cite{cicalese2019new} is polynomial, however, no precise complexity is constructed. Since the time for projection the original data to a 2-dimension data is $NM$ and the time of finding the optimal partition in 2-dimensional space using the method in \cite{kurkoski2014quantization} is $O(M^3)$, the time complexity of the algorithm in \cite{cicalese2019new} should be at least $O(NM+M^3)$. 

Based on the well-known SMAWK algorithm \cite{aggarwal1987geometric}, we show that the time complexity of the algorithm in \cite{cicalese2019new} can be further reduced from $O(M^3)$ to $O(M \log M)$. Indeed, the SMAWK algorithm can be applied to reduce the time complexity of algorithm in \cite{kurkoski2014quantization} to $O(KM)$ if the binary data is ordered (see \cite{iwata2014quantizer} and \cite{he2019dynamic} for detail).
However, to order a data of size $M$, the fastest sorting technique requires  $O(M \log M)$. Thus, the problem in \cite{kurkoski2014quantization} can be solved in $O(M \log M)$ that reduces the polynomial time complexity in step (2)  of algorithm in \cite{cicalese2019new} to $O(M \log M)$. The total time complexity of the proposed algorithm in \cite{cicalese2019new}, therefore, is $O(NM+M^3)$. 

\subsection{ Jensen's Inequality}
\label{apd: jensen inequality}

Jensen inequality states  that for a random variable $T$, then
$\mathbb{E}[f(T)] \geq  f(\mathbb{E}[T])$ if $f(x)$ is convex, and 
$\mathbb{E}[f(T)] \leq  f(\mathbb{E}[T])$  if $f(x)$ is concave.

Now, let $T \in \{t_1,t_2,\dots,t_K\}$ be a random variable with the uniform distribution $\dfrac{1}{K}$. If $f(x)$ is concave,  using Jensen's inequality:
\begin{equation*}
  f(\dfrac{\sum_{i=1}^{K}t_i}{K}) \geq \sum_{i=1}^{K}\dfrac{1}{K}f(t_i), 
\end{equation*}
which is equivalent to:
\begin{equation}
\label{eq: appendix b}
  K f(\dfrac{\sum_{i=1}^{K}t_i}{K}) \geq \sum_{i=1}^{K}f(t_i).  
\end{equation}
Thus, (\ref{eq: 15}) is a direct result of (\ref{eq: appendix b}) using $t_i=p(x_i|z_j)$ and $K=N-1$. 

\subsection{Fano's Inequality}
\label{apd: fano}

If the impurity function is entropy  i.e.,  $f(x)=-x\log x$, then: 
\begin{align*}
    I_Q&=\sum_{k=1}^{K} \sum_{i=1}^{N} p(z_k) \Big(-p(x_i|z_k) \log (p(x_i|z_k))\Big)\\
    &=\sum_{k=1}^{K} p(z_k)H(X|z_k)=H(X|Z).
\end{align*}
By plugging $f(x)=-x\log x$ into (\ref{eq: 18}):
\begin{eqnarray}
H(X|Z) &\leq& -e_Q \log e_Q - (N-1)\dfrac{1-e_Q}{N-1} \log \dfrac{1-e_Q}{N-1} \nonumber\\
&=&-[e_Q \log e_Q +(1-e_Q) \log(1-e_Q)]\nonumber\\&+& (1-e_Q) \log (N-1) \label{eq: 19}\\
&=& H(1-e_Q) + (1-e_Q) \log (N-1),  \label{eq: 20}
\end{eqnarray}
with (\ref{eq: 19}) is due to a bit of algebra, (\ref{eq: 20}) is due to the binary entropy function is symmetric, i.e.,  $H(e_Q)=H(1-e_Q)=-[e_Q \log e_Q +(1-e_Q) \log(1-e_Q)]$. 

Let us consider $X$ and $Z$ as two random variables that represent the input and the output symbols of a communication channel. Errors might occur during transmissions.   Suppose that the receiver estimates the transmitted symbol based on the received $z_k$ as $x_{k^*}$ where $k^*=\argmax_{i}  p(x_i|z_k)$ (maximum likelihood decoding). Thus, the error probability of this decoding scheme is $P_e=1-\sum_{k=1}^{K}p(z_k)p(x_{k^*}|z_k)=1-e_Q$.  Then, 
\begin{eqnarray*}
H(X|Z) &\leq& H(1-e_Q)+ (1-e_Q) \log (N-1) \\
&=& H(P_e)+ P_e \log (|X|-1),
\end{eqnarray*}
which is identical to the well-known Fano's inequality \cite{cover2012elements}. 

\subsection{ Boyd-Chiang's Upper Bound of Channel Capacity}
\label{apd: boyd-chiang}
The mutual information $I(X;Z)$ between channel input and channel output is defined by $I(X;Z)=H(X)-H(X|Z)$.  However, from Theorem \ref{theorem: 2}, if the impurity function is entropy i.e., $f(x)=-x \log (x)$ and $l(x)=-\log (x)$ then $H(X|Z)=I_Q \geq l(e_Q)=-\log (e_Q)$. Now, by using the uniform input distribution,
\begin{eqnarray}
I(X;Z)&=&H(X)-H(X|Z)\nonumber\\
&\leq& H(\dfrac{1}{N},\dfrac{1}{N},\dots,\dfrac{1}{N})\nonumber\\&-&l(\sum_{k=1}^{K}p(z_k)p(x_{k^*}|z_k))\label{eq: ha 5}\\
&=& \log N - l(\sum_{k=1}^{K}p(x_{k^*}) p(z_k|x_{k^*}))\label{eq: ha 6}\\
&=&\log N - l(\dfrac{1}{N}\sum_{k=1}^{K}p(z_k|x_{k^*}))\label{eq: ha 7}\\
&=&\log N+\log (\dfrac{\sum_{k=1}^{K}p(z_k|x_{k^*})}{N})\label{eq: ha 8}\\
&=& \log (\sum_{k=1}^{K}p(z_k|x_{k^*})), \label{eq: ha 9}
\end{eqnarray}
with (\ref{eq: ha 5}) is due to $H(X|Z) \geq l(e_Q)$ and $e_Q=\sum_{k=1}^{K}p(z_k)p(x_{k^*}|z_k)$, (\ref{eq: ha 6}) is due to Bayes's theorem, (\ref{eq: ha 7}) is due to the input distribution is uniform, (\ref{eq: ha 8}) is due to $l(x)=-\log (x)$, (\ref{eq: ha 9}) is due to a bit of algebra. 

Let us now consider $X$ and $Z$ as two random variables that represent the input and the output symbols of a communication channel. Due to the errors during transmissions,  a channel matrix $A$ whose entry $A_{ij}=p(z_j|x_i)$ denotes the probability of the transmitter transmitted symbol $x_i$ but the receiver decoded to symbol $z_j$.  Now, since the input distribution is uniform, from $p(z_k)p(x_{k^*}|z_k)=p(x_{k^*})p(z_k|x_{k^*})$, then $p(z_k|x_{k^*})$ is the largest entry in  ${k^*}^{th}$ column of channel matrix. Thus, the upper bound of channel capacity is $\log (\sum_{k=1}^{K}p(z_k|x_{k^*}))$ that is identical to the bound constructed by Boyd and Chiang \cite{chiang2004geometric}. 

\subsection{Proof of Theorem \ref{theorem: monotonic decreasing of upper bound}}
\label{apd: proof of theorem monotonic upper bound}
\begin{proof} We show that $u(e_Q) = f(e_Q)+(N-1)f(\dfrac{1-e_Q}{N-1})$ is a non-increasing function. 
\begin{eqnarray}
u'(e_Q)&=&f'(e_Q)-(N-1)\dfrac{1}{N-1}f'(\dfrac{1-e_Q}{N-1})\\
&=&f'(e_Q)-f'(\dfrac{1-e_Q}{N-1}).
\end{eqnarray}
Since $k^*=\argmax_i p(x_i|z_k)$ and $\sum_{i=1}^{N}p(x_i|z_k)=1$,  $p(x_{k^*}|z_k) \geq \dfrac{1}{N}$. Thus, $$e_Q=\sum_{k=1}^{K}p(z_k)p(x_{k^*}|z_k) \geq \sum_{k=1}^{K}p(z_k) \dfrac{1}{N}=\dfrac{1}{N}.$$
Therefore, $1-e_Q \leq \dfrac{N-1}{N}$. Thus,
\begin{equation}
    e_Q \geq \dfrac{1}{N} \geq \dfrac{1-e_Q}{N-1}.
\end{equation}
Now since $f'(e_Q)$ is a non-increasing function due to $f(e_Q)$ is concave. Therefore, 
\begin{equation}
    u'(e_Q)=f'(e_Q)-f'(\dfrac{1-e_Q}{N-1}) \leq 0.
\end{equation}
Or, $u(e_Q)$ is a non-increasing function. 

Finally, it is possible to verify that if $e_Q=\dfrac{1}{N}$ or $e_Q=1$, then the upper bound is tight i.e.,  $u(e_Q)=I_Q$ for both the entropy impurity and the Gini index impurity. Indeed, if $e_Q=\dfrac{1}{N}$,  $p(x_{i}|z_k)=\dfrac{1}{N}$ $\forall i,k$ and then $u(e_Q)=I_Q=N f(\dfrac{1}{N})$. If $e_Q=1$, $p(x_{k^*}|z_k)=1$ and  $p(x_{i}|z_k)=0$ $\forall$ $i \neq k^*$ and then $u(e_Q)=I_Q=0$. 
\end{proof}

\subsection{Proof of Theorem \ref{theorem: p(x,y)}}
\label{apd: proof of theorem {theorem: p(x,y)}}

\subsubsection{Proof of Theorem \ref{theorem: p(x,y)}-(a)}
\begin{proof} We first consider the case when $K = N$, we show that $e_{Q_{e^{\max}}} \geq e_Q, \forall Q$. We have:
\begin{align*}
e_{Q_{e^{\max}}} &= \sum_{j^*=1}^K{p(z_{j^*}) \max_i{p(x_i|z_{j^*})}} \\  
    &= \sum_{j^*=1}^K{\max_i{p(x_i, z_{j^*})}} \\
    &= \sum_{j^*=1}^K{\sum_{j:Q(y_j) = z_{j^*}}{\max_i p(x_i,y_j)}} \\
    & \geq  \sum_{j=1}^K{\sum_{j:Q(y_j) = z_j}{p(x_i,y_j)}}  \\
    &=  e_{Q}.
\end{align*} 
Note that $Q(y_j) = z_j$ in the index of the sum in the last equation represents any arbitrary partition scheme.

We now show that if a quantizer $Q$  produces $e_{\max} = \max_Q {e_Q}$ then it must has the structure of $Q_{e^{\max}}$.
We will prove this by contradiction.  Suppose that a quantizer $Q$ produces the partitions $z_1,z_2,\dots,z_K $ that has $e^{\max}$, but there exists a  $y_n$ that is partitioned to $z_{l}$, such that $l \neq \argmax_{1\leq i \leq N}  p(x_i,y_n)$. Let $m =  \argmax_{1\leq i \leq N} p(x_i,y_n)$. Now, let consider a quantizer $Q'$ which is constructed from quantizer $Q$ by moving $y_n$ from $z_{l}$ to $z_{m}$. This new quantizer $Q'$ produces a new partition $\{z'_1,\dots,z'_{l},\dots ,z'_{m}, \dots, z'_K\}$ with $z'_k = z_k$, $\forall$ $k \neq l, m$, with corresponding $p'(x_i, z_k)$ and $p'(x_i| z_k)$. 

From the definition of $e_Q$, we have:
\begin{small}
\begin{align*}
    e_Q - e_Q'&= p(z_l)p(x_{l*}| z_l)+p(z_m)p(x_{m^*}|z_m)\\&- p'(z_l)p'(x_{l*}| z_l)- p'(z_m)p'(x_{m^*}|z_m)\\
   &= p(x_{l*},z_l)+p(x_{m^*},z_m) - p'(x_{l*}, z_l)- p'(x_{m^*},z_m)\\
		&=\sum_{j:Q(y_j) = z_l}{p(x_{l*}, y_j)} + \sum_{j:Q(y_j) = z_m}{p(x_{m*}, y_j)} \\
		&- \sum_{j:Q(y_j) = z_l}{p'(x_{l*}, y_j)} -\sum_{j:Q(y_j) = z_m}{p'(x_{m*}, y_j)}\\	
		&= p(x_{l*}, y_n) - p(x_{m*},y_n).
\end{align*}
\end{small}Since by assumption that $p(x_{m},y_n) > p(x_l, y_n)$, we have $e_{Q'} < e_{Q}$ which is a contradiction.  Thus,  any partition scheme that achieves $e^{\max}$ must have the structure of maximum likelihood of $Q_{e^{\max}}$.
\end{proof}

\subsubsection{Proof of Theorem \ref{theorem: p(x,y)}-(b)}
\begin{proof} Theorem \ref{theorem: p(x,y)}-(a) handled the case when $K = N$ and showed that any partition scheme that achieves $e^{\max}$ must have the structure of maximum likelihood of $Q_{e^{\max}}$. On the other hand, Theorem \ref{theorem: p(x,y)}-(b) finds the partition that achieves $e^{\max}$ when $K>N$. Interestingly, we show that the mapping of $Q_{e^{\max}}$ in Theorem \ref{theorem: p(x,y)}-(a) that partitions the data to $N$-nonempty partitions and $K-N$ empty partitions still produces $e^{\max}$. Let $j^* = \argmax_i p(x_i, y_j)$  and define quantizer $Q_{e^{\max}}$  with the following structure:
\begin{equation}
\label{eq: mapping k>n}
Q_{e^{\max}}(y_j) = z_{j^*},
\end{equation}
then $Q_{e^{\max}}$ produces $e^{\max} = \max_Q{e_Q}$ even if $K>N$. Moreover,  due to the mapping in (\ref{eq: mapping k>n}), $Q_{e^{\max}}$ produces $N$ nonempty partitions and $K-N$ empty partitions. 

Indeed, suppose the quantizer $Q_{e^{\max}}$ produces $K'$ nonempty partitions and $K \geq K'>N$. We show that there exists another quantizer $Q$ having exactly $N$ nonempty partitions which produces the same  $e^{\max}$ as $Q_{e^{\max}}$. Now, since $K'>N$, there exist two partitions $z_i,z_j$ such that: $$i^*=j^*=\argmax_{1 \leq t \leq N} p(x_t|z_i)=\argmax_{1 \leq t \leq N} p(x_t|z_j).$$

Next, consider a new quantizer $Q$ that maps two partitions $z_i$ and $z_j$ into a single partition $z_{k}$, we show that $Q$ still provides the same $e^{\max}$ as $Q_{e^{\max}}$.  Indeed, since $i^*=j^*$ and $z_i \cup z_j =z_k$, $z_i \cap z_j =\emptyset$, we have:
$$ i^*=j^*=k^*=\argmax_{1 \leq t \leq N} p(x_t|z_k),$$and
$$p(x_{k^*},z_k)=p(x_{i^*},z_i) + p(x_{j^*},z_j).$$Thus,
\begin{small}
\begin{eqnarray}
p(z_k) p(x_{k^*}|z_k)&\!=\!&p(x_{k^*},z_k)=p(x_{i^*},z_i) + p(x_{j^*},z_j)\\
&\!=\!& p(z_i)p(x_{i^*}|z_i) + p(z_j)p(x_{j^*}|z_j). 
\end{eqnarray}
\end{small}

By definition of $e_Q$ in (\ref{eq:i^*}) and noting that $Q$ is identical to $Q_{e^{\max}}$ except that two partitions $z_i$ and $z_j$ are grouped into a single partition $z_{k}$, $e_Q=e^{\max}$. By induction method, after at most $K'-N$ times grouping, there exist a quantizer $Q$ having exactly $N$ nonempty partitions which still produce $e^{\max}$. Moreover, this quantizer satisfies the mapping in (\ref{eq: mapping k>n}). 
\end{proof}

\subsection{Proof of Theorem \ref{theorem: 8}}
\label{apd: proof of theorem 8}
\begin{proof}
For the Gini index impurity function, $f(x)=x(1-x)$ and $l(x)=1-x$. Thus, 
\begin{eqnarray}
R(e^{\max})&=&\dfrac{f(e^{\max})+(N-1)f(\dfrac{1-e^{\max}}{N-1})}{l(e^{\max})}\nonumber\\
&=&\dfrac{e^{\max}(1-e^{\max})}{1-e^{\max}} \nonumber\\&+&\dfrac{(N-1) \dfrac{1-e^{\max}}{N-1}(1-\dfrac{1-e^{\max}}{N-1})}{1-e^{\max}} \label{eq: v1}\\
&=& e^{\max}+ 1-\dfrac{1-e^{\max}}{N-1} \label{eq: v2}\\
&\leq & e^{\max}+1 \label{eq: v3}\\
&\leq & 2 \label{eq: v4},
\end{eqnarray}
with (\ref{eq: v1}) due to $f(x)=x(1-x)$ and $l(x)=1-x$, (\ref{eq: v2}) and (\ref{eq: v3}) due to a bit of algebra, (\ref{eq: v4}) due to $e^{\max} \leq 1$. Noting that one can use $e^{\max}+1$ as another approximation for Gini index impurity. 
\end{proof}

\subsection{Proof of Theorem \ref{theorem: 9}}
\label{apd: proof of theorem 9}

\begin{proof}
For entropy impurity, $f(x)=-x\log(x)$ and $l(x)=-\log(x)$, plug in the upper bound and the lower bound in Theorem \ref{theorem: 1} and Theorem \ref{theorem: 2}, we have:
\begin{equation}
    R(e^{\max})=\dfrac{H(e^{\max})+(1-e^{\max})\log(N-1)}{-\log(e^{\max})}.
\end{equation}
Since $\log(N-1) < \log N$, we have:
\begin{eqnarray}
\label{eq: 201}
    R(e^{\max})&=&\dfrac{H(e^{\max})+(1-e^{\max})\log(N-1)}{-\log(e^{\max})}\nonumber\\&<& \dfrac{H(e^{\max})+(1-e^{\max})\log N}{-\log(e^{\max})}.
\end{eqnarray}
To prove Theorem \ref{theorem: 9}, we want to show that the inequality below holds.
\begin{equation}
\label{eq: 300}
    \dfrac{H(e^{\max})+(1-e^{\max})\log N}{-\log(e^{\max})} \leq \log^2 N, \forall N \geq N^{\min}.
\end{equation}
 This is equivalent to show that: 
 \begin{small}
 \begin{equation*}
    \log^2 N(\!-\!\log(e^{\max}))\!-\!(1\!-\!e^{\max})\log N\!-\!H(e^{\max}) \geq 0, \forall N \geq N^{\min}. 
 \end{equation*}
 \end{small}Indeed, using a bit of algebra,
\begin{small}
\begin{eqnarray*}
&&\log^2 N(-\log(e^{\max}))-(1-e^{\max})\log N-H(e^{\max})\\
&=&-\!\log(e^{\max})[\log^2 \!N \!- \!2\log \!N \dfrac{1-e^{\max}}{2(\!-\!\log(e^{\max}))}\!\\&+&\! (\dfrac{1-e^{\max}}{2(\!-\!\log(e^{\max}))})^2\!-\!\dfrac{H(e^{\max})}{\!-\!\log(e^{\max})}
\!-\!(\dfrac{1-e^{\max}}{2(\!-\!\log(e^{\max}))})^2] \\
&=& -\log(e^{\max})[(\log N - \dfrac{1-e^{\max}}{2(-\log(e^{\max}))})^2 \\&-& \dfrac{4 H(e^{\max})(-\log(e^{\max}))+(1-e^{\max})^2}{(-2 \log(e^{\max}))^2}]\\
&=& -\log(e^{\max})[(\log N \!-\! \dfrac{1-e^{\max}}{2(-\log(e^{\max}))})^2 \!\\&-&  \! (\dfrac{\sqrt{4 H(e^{\max})(-\log(e^{\max}))+(1-e^{\max})^2}}{-2 \log(e^{\max})})^2].
\end{eqnarray*}
\end{small}Now, if: 
\begin{eqnarray}
    \log N &\geq& \dfrac{1-e^{\max}}{-2\log (e^{\max})} \nonumber\\&+& \dfrac{\sqrt{4 H(e^{\max})(-\log(e^{\max}))+(1-e^{\max})^2}}{-2 \log(e^{\max})}\nonumber\\&=&S(e^{\max}),
\end{eqnarray}
then (\ref{eq: 300}) holds. Thus, $R(e^{\max}) < \log^2 N$ holds if $N \geq 2^{S(e^{\max})}=N^{\min}$. 
\end{proof}

\end{document}